\documentclass[acmsmall,nonacm]{acmart}
\usepackage{tikz}
\usepackage{titleref}

\usepackage [electronic,xcolor,hyperref,notion,quotation]{knowledge}

\knowledge{notion}
| Sparse Boolean Matrix Multiplication Hypothesis

\knowledge{notion}
| Triangle Hypothesis

\knowledge{notion}
| Hyperclique Hypothesis

\knowledge{notion}
| Strong Exponential Time Hypothesis

\knowledge{notion}
| 3SUM Hypothesis

\knowledge{notion}
| Min-Weight-$k$-Clique Hypothesis

\knowledge{notion}
| Zero-$k$-Clique Hypothesis

\knowledge{notion}
| Combinatorial $k$-Clique Hypothesis

\newcommand{\datarule}{{\,:\!\!-\,}} %

\newtheorem{hypothesis}{Hypothesis}

\newcommand{\softO}{\tilde{O}}

\title{Lower Bounds for Conjunctive Query Evaluation} 

\author{Stefan Mengel}
\affiliation{		
    \institution{Univ.~Artois, CNRS, Centre de Recherche en Informatique de Lens (CRIL)}
    \city{Lens}
    \country{France}}
\email{mengel@cril.fr}



\begin{document}
    
    \begin{abstract}
        In this tutorial, we will survey known results on the complexity of conjunctive query evaluation in different settings, ranging from Boolean queries over counting to more complex models like enumeration and direct access. A particular focus will be on showing how different relatively recent hypotheses from complexity theory connect to query answering and allow showing that known algorithms in several cases can likely not be improved.
    \end{abstract}

    \maketitle
    
\section{Introduction}\label{sec:introduction}

Understanding the complexity of different forms of query answering is one of the major research areas in database theory. Historically, the complexity of queries has been studied in three ways: in combined complexity, in data complexity~\cite{Vardi82}, or in parameterized complexity~\cite{Grohe02}. When trying to understand the necessary runtime to answer queries, combined complexity is very pessimistic: since the query and the database are both the inputs, answering first-order queries is PSPACE-hard and even conjunctive queries are NP-hard~\cite{ChandraM77}. On the other end of the spectrum, in data complexity all first-order queries are classified as in PTIME and thus easy~\cite{Vardi82}, which blurs the different complexities of different queries. So in a sense, both of these classical approaches are not completely satisfying when trying to assess the hardness of concrete queries.

A more refined approach is parameterized complexity~\cite{Grohe02}: similarly to data complexity, one considers the database as the main input, but instead of neglecting the query completely, one sees its size as a parameter. The runtime dependence on the parameter may be bad, but crucially it should not appear in the exponent of the input size, so roughly, for parameter $k$ and input size $n$, a runtime of $2^k n$ is good, but $n^k$ is bad. Parameterized complexity has led to many interesting results on the complexity of query answering, in particular the characterization of all \emph{classes} of conjunctive queries that can be solved efficiently from the perspective of parameterized and combined complexity~\cite{GroheSS01,Grohe07}. However, despite these merits, parameterized complexity has inherent limitations: it is only applicable to \emph{classes} of queries, and thus we cannot use it to learn something about the complexity of specific, concrete queries.

Fine-grained complexity opens up the possibility of attacking this question and determining the exact exponents of optimal query answering algorithms for individual queries. The idea of the area is to develop a complexity theory for problems inside PTIME to determine optimal runtimes for problems that are in classical theory all considered easy. The goal is to complement problems with algorithms that have runtime $O(n^{c})$ with lower bounds that rule out any algorithm with runtime $O(n^{c-\varepsilon})$ for all $\varepsilon > 0$.

Of course, with today's techniques, one cannot expect to show unconditional lower bounds. Also, it turns out that common classical assumptions like P $\ne$ NP or even the Exponential Time Hypothesis, which plays a major role in parameterized complexity, appear to not be sufficient. Instead, fine-grained complexity relies on its own set of different hypotheses. Often these posit that known algorithms for important problems that have, despite considerable efforts, not been improved for a long time, are optimal. So the results of fine-grained complexity theory are generally of the form ``If there is no better algorithm for problem X, then there is no better algorithms for problem Y, either.'' where X can be one of several different, well-studied problems and $Y$ is the problem that one is interested in. This makes fine-grained complexity a less coherent theory than, say, the theory of NP-completeness, and the confidence in its main hypotheses is certainly lower. However, it still provides techniques to give insights into which runtimes improvements would imply major, surprising algorithmic advances on well-known problems.

This paper is \emph{not} a survey of or introduction to fine-grained complexity itself---there are certainly others better qualified to do so\footnote{For example, there will be a keynote by Virginia Vassilevska Williams on this topic at PODS this year.} and this type of work already exists, see e.g.~\cite{VassilevskaWilliams2018} or the videos from a recent Simons Institute workshop that explain the area to a database theory audience~\cite{Simons23}. Instead, this paper will be limited to a single question: how can fine-grained complexity be used to determine the complexity of conjunctive query answering. Here query answering is not restricted to materializing the complete query result or solving Boolean queries, but also covers tasks like counting (as a prototypical aggregation task), enumeration and direct access. A large part of this paper will be on linear time bounds, as this is the case that is by far most studied. As we will see, there is a wealth of results on this. Afterwards, we will turn to super-linear runtimes. Unfortunately, lower bounds for this case are far less developed, but we will give an overview of what is known and which techniques exist. 
Throughout the whole paper, we will introduce important hypotheses from fine-grained complexity that have been proven useful in database theory and show how they have been applied.

\paragraph*{Acknowledgements} This paper was written to complement a tutorial that the author was invited to give at PODS 2025. The authors thanks the program committee for this opportunity to collect and present the results surveyed in this paper.

Florent Capelli and Christoph Berkholz provided thoughtful comments, advice and discussion on an earlier version of this paper. The author is grateful for the time and effort they invested into this, which greatly improved the presentation of this survey.

\section{Preliminaries}\label{sec:preliminaries}

\subsection{Database Queries}
We assume that the reader is familiar with the basics of database theory, in particular the relational model, see e.g.~the first part of~\cite{ArenasBLMP21}.
In this paper, we will focus only on \emph{conjunctive queries}, so queries of the form $q(X) \datarule R_1(X_1), \ldots, R_\ell(X_\ell)$, where $X$ and the $X_i$ are sets of variables with $X \subseteq \bigcup_{i \in [\ell]} X_i$. If $X = \bigcup_{i \in [\ell]} X_i$, then we call $q$ a \emph{join query}. If $X=\emptyset$, we call $q$ \emph{Boolean}. If all relation symbols $R_i$ in $q$ are different, we say that $q$ is \emph{self-join free}; if the same relation symbol repeats, we say that $q$ has self-joins. We assume basic familiarity with the AGM-bound~\cite{AtseriasGM13} and worst-case optimal join algorithms~\cite{Ngo18}, see also~\cite[Chapters 27 and 28]{ArenasBLMP21}. However, these concepts will not be central, and knowledge of their existence will be enough.

As is standard, we assign hypergraphs to queries: for the query $q(X)\datarule R_1(X_1), \ldots, R_\ell(X_\ell)$, the corresponding hypergraph $H=(V,E)$ has vertices $V:= \bigcup_{i \in [\ell]} X_i$ and the edge set $\{X_i\mid i \in [\ell]\}$.
A hypergraph $H= (V,E)$ is called \emph{acyclic} if the following process ends with a hypergraph with no vertices or edges: while possible, delete a vertex that is contained in at most one edge, or delete an edge that is a subset of another edge. A conjunctive query is called acyclic if its hypergraph is acyclic. Acyclic queries are important and well studied because of their good algorithmic properties, see e.g.~\cite[Chapters 20 and 21]{ArenasBLMP21}. For background and many properties of acyclic hypergraphs, see~\cite{BeeriFMY83,Brault-Baron16,BerkholzGS20}. We say that a hypergraph is \emph{$h$-uniform} for some $h\in \mathbb N$ if all of its edges have size $h$.

All runtimes for query answering that we give will be for fixed queries and thus have no dependence on the size of the query (but will of course depend on the actual query; we want to differentiate hard queries from easy ones after all). We generally denote the size of the input database, i.e.~the overall number of tuples, by $m$ while $n$ is generally used for the size of the domain.

\subsection{Fine-Grained Complexity}

As said in the introduction, fine-grained complexity is concerned with finding the optimal exact exponent of runtime bounds for problems that can be solved in polynomial time. The considered machine model are standard random access machines with logarithmic word-size and unit cost operations, see e.g.~\cite{GrandjeanJ22} for a detailed discussion of this model.

We will mostly be interested in lower bounds saying that for a constant $c$ there is no $\varepsilon > 0$ such that a problem can be solved in time $O(n^{c-\varepsilon})$. These lower bounds will be based on different hypotheses which we will introduce throughout the paper. Since we mostly ignore sub-polynomial factors, we will give most runtime bounds in $\softO$-notation, which is analogous $O$-notation, but ignores factors of the form $n^{o(1)}$~\footnote{In the literature, $\softO$ often only ignores polylogarithmic instead of subpolynomial factors. However, since in fine-grained complexity one generally ignores all subpolynomial factors, our slightly more liberal definition will be more convenient.}

While we will not see this here, many algorithms and reductions in fine-grained complexity are randomized, and they might fail with a probability $p< \frac{1}{2}$. The probability $p$ can be decreased in a standard fashion by repeatedly running the algorithm. All our reductions here will be deterministic, but it is still good to keep randomization in mind when reading papers in fine-grained complexity. In particular, all hypotheses we mention here are in general also made in this randomized model, sometimes without mentioning it explicitly.

\subsection{The Complexity of Matrix Multiplication}

Fast matrix multipliation is one of the work horses of efficient algorithms in fine-grained complexity, and it also increasingly plays a major role in query evaluation~\cite{AmossenP09,AboKhamisHS24,Hu24}. So let us discuss this basic operation a little. 

With the naive algorithm, one can multiply two $n\times n$-matrices in time $O(n^3)$, which was thought to be optimal until the breakthrough work of Strassen~\cite{Strassen69}, showing that the problem has a sub-cubic algorithm. Since then there has been a race to improve the runtime bounds.
The decisive quantity here is the \emph{matrix multiplication exponent} $\omega$ which is the smallest real number such that there is an algorithm that computes the product of two $n\times n$-matrices with $n^{\omega + o(1)}$ operations\footnote{Note the following slight subtlety: $\omega$ is defined as a limit, and it is known that there is no single algorithm that attains this limit~\cite{CoppersmithW82}. So whatever the value of $\omega$ is, there is no algorithm with runtime $O(n^\omega)$. However, since we will use our version of the $\softO$-notation that does not take into account factors of the form $n^{o(1)}$, we ignore this fact in the remainder of this paper.}. From Strassen's algorithm and the fact that we need to compute $n^2$ entries of the output, it follows that $2\le \omega < 3$ and, while the exact number is not known, there have been regular improvements of the upper bound in recent years, the current record being $\omega < 2.371339$~\cite{AlmanDWXXZ25}.

The best known lower bound for matrix multiplication is quadratic in $n$~\cite{Blaser03}, and it does not seem to be a widely held assumption that $\omega > 2$. In fact, many works introducing algorithms that use matrix multiplication as a sub-routine and whose runtime thus depends on $\omega$ in more or less complicated ways, also give resulting runtimes for the case $\omega = 2$. That said, while there has been steady progress over decades since~\cite{Strassen69}, it appears likely that we are far away from determining $\omega$.

We will mostly be concerned with \emph{Boolean matrix multiplication}, i.e.~matrix multiplication over the Boolean semi-ring where, given two input matrices $A=(a_{ij})$, $B=(b_{ij})$, both from $\{0,1\}^{n\times n}$ we want to compute the matrix $C=(c_{ij})$ with $c_{ij} := \bigvee_{k\in [n]} a_{ik} \land b_{kj}$. So far, the best algorithms for Boolean matrix multiplication in fact use multiplication over the reals: given $A$ and $B$, interpret their entries as real numbers and multiply them over the reals. Then, it is easy to see that substituting any non-zero entry of the output $C$ by $1$ gives the result of Boolean matrix multiplication. Note that the runtime of this approach is $\softO(n^\omega)$, so the same as that of non-Boolean matrix multiplication.

Of particular interest to us is a setting often called \emph{sparse} Boolean matrix multiplication where the runtime is not measured in the dimension $n$ but in the number $m$ of non-zero entries of the input and output. Note that if there is an algorithm of runtime $\softO(m)$ then $\omega = 2$.
However, the other way round, even if $\omega =2$, it does not follow that sparse matrix multiplication has a (near)-linear-time algorithm,  since the input and output can have far less than $n^2$ non-zero entries. In that case, algorithms for dense matrix multiplication and $\omega =2$ are not obviously helpful. That said, there are several algorithms for sparse matrix multiplication that use algorithms for dense matrix multiplication in more subtle ways: the current best runtime is an algorithm for sparse Boolean matrix multiplication with runtime $O(m^{1.3459})$ which in case $\omega = 2$ is even $O(m^{4/3 + \varepsilon})$ for every $\varepsilon > 0$~\cite{AbboudBFK24}. The general belief in fine-grained complexity is that $O(m^{4/3})$ can likely not be beaten as a runtime, see again~\cite{AbboudBFK24}, and that in particular there is no linear-time algorithm for sparse Boolean matrix multiplication.

\AP \begin{hypothesis}[\intro {Sparse Boolean Matrix Multiplication Hypothesis}]
    There is no algorithm that solves sparse Boolean matrix multiplication in time $\softO(m)$.
\end{hypothesis}

\section{Understanding Linear-Time Query Answering} \label{sec:linear}

In this section, we will turn to the complexity of different query answering problems. Since the case of (nearly) linear time is the best understood, we will focus on understanding it first. In Section~\ref{sct:superlinear}, we will turn to the less well understood superlinear case. We will start our considerations with Boolean queries before considering more complex questions like enumeration, counting and direct access in later subsections.

\subsection{Boolean Queries}

We first consider Boolean queries, since any lower bound that we can show for them will in particular also give us lower bounds for other question like counting or enumeration. It will turn out that in this whole section, acyclic queries will play a major role, due to the following classical result by Yannakakis~\cite{Yannakakis81}.

\begin{theorem}\label{thm:yannakakis}
    For every acyclic Boolean conjunctive query $q$ there is an algorithm, given a database~$D$ decides in linear time if $q$ it true on $D$.
\end{theorem}
For a streamlined proof of Theorem~\ref{thm:yannakakis}, see~\cite[Chapter~21]{ArenasBLMP21}. In the remainder of this subsection, we will explain why acyclic queries are likely the \emph{only} Boolean queries that can be solved in linear time and which hypotheses from fine-grained complexity suggest this.

\subsubsection{Graphlike Queries}\label{sct:boolean}

It will be useful to first consider the case in which all atoms have arity two, so where the hypergraph of the query is a graph. Acyclicity of the query is then simply the usual notion of acyclicity of graphs, so the absence of any cycles. 

We want to make it plausible that cyclic graphlike queries cannot be solved in linear time, so it appears useful to first consider queries that consist only of a single cycle and, to make things even easier, we for now consider only the triangle query which has been studied extensively in the database literature, see e.g.~the pointers in~\cite{Ngo18}. We write it here as
\begin{align*}
    q^\triangle() \datarule R_1(x,y), R_2(y,z), R_3(z,x).
\end{align*}

What can we say about the complexity of $q^\triangle$? First, if we consider the corresponding join query
\begin{align*}
    \bar{q}^\triangle(x,y,z) \datarule R_1(x,y), R_2(y,z), R_3(z,x),
\end{align*}
then we see by applying the celebrated AGM-bound~\cite{AtseriasGM13} that $|\bar q(D)| \le m^{3/2}$ for any database of size $m$, and using a worst-case optimal join algorithm~\cite{Ngo18}, we can compute all answers in time $\softO(m^{3/2})$. So in particular, $q^\triangle$ can be solved in time $\softO(m^{3/2})$. But can we do better? For $\bar q^\triangle$ we can certainly not: the AGM-bound is tight, so there are inputs with $\Omega(m^{3/2})$ output tuples, and since we have to return them all, we cannot do so in time $m^{3/2-\varepsilon}$ for any $\varepsilon>0$. However, if we only want to know if there is an answer, we can use matrix multiplication to speed up query answering, as shown by Alon, Yuster and Zwick~\cite{AlonYZ97}. Since the proof if short but quite instructive, we give it here.

\begin{theorem}\label{thm:AYZ}
    If there is an algorithm for Boolean matrix multiplication of $n\times n$-matrices with runtime $\softO(n^\omega)$, then there is an algorithm that answers $q^\triangle$ on inputs of size $m$ in time $\softO(m^{\frac{2\omega}{\omega+1}})$.
\end{theorem}
\begin{proof}
    We use a technique based on degree splits. To this end, define the degree of an element of the active domain as the number of tuples in which it appears in the input database. Then we call a domain element light if it has degree at most $\Delta := m^{\frac{\omega-1}{\omega + 1}}$ and heavy otherwise. We will first show how to detect answers of $\bar q^\triangle$ that contain a light element, say, for the variable $y$. To do so, we first compute the answers of $q_y(x,y,z)\datarule R_1(x,y), R_2(y,z)$ in which $y$ takes a light element. We can do this by first filtering out all heavy elements at the position of $y$ in the relations $R_1$ and $R_2$, then we take all remaining tuples in $R_1$ and extend them in the up to $\Delta$ ways to $z$. Certainly, this can be done in time $\softO(m\Delta)$. In the final step, we can filter with the relation $R_3$ to check if any of the candidates thus computed is an answer to $\bar q^\triangle$. Analogously, we can check if there is an answer to $\bar q^\triangle$ that takes a light value in $x$ or $z$.
    
    It remains to check for answers that only contain heavy values. There are at most $m/\Delta$ such values and we can compute them efficiently, so we only have to show how to solve $q^\triangle$ on a database with up to $n':= m/\Delta$ domain elements. So let $R_1'$ and $R_2'$ be the relations we get from $R_1$ and $R_2$, respectively, by restriction to heavy values. We first solve the query $q'(x,z)\datarule R_1'(x,y), R_2'(y,z)$ by matrix multiplication: let $A$ be the adjacency matrix of $R_1'$, so the Boolean $n'\times n'$-matrix indexed by the active domain which contains a $1$ in an position $(a,b)$ if and only if $(a,b)\in R_1'$. Define $B$ analogously as the adjacency matrix of $R_2'$. Let $C$ be the result of Boolean matrix multiplication of $A$ and $B$, then $C$ is the adjacency matrix of the relation containing the answers to $q'$ and thus we can read off the answers directly from $C$. Filtering by $R_3$ again, lets us check if $q^\triangle$ has an answer.

    The overall runtime of the algorithm is $\softO(m\Delta) = \softO(m^{\frac{2\omega}{\omega+1}})$ for the first part of the algorithm and $\softO(n'^\omega) = \softO(\left(\frac{m}{\Delta}\right)^\omega)= \softO(m^{\frac{2\omega}{\omega+1}})$ for the second part, which gives the desired runtime bound.
\end{proof}

For $\omega > 0$, the mapping $\omega \mapsto \frac{2\omega}{\omega + 1}$ is monotone and increasing in $\omega$, so better matrix multiplication algorithms give better algorithms for the triangle query. If $\omega = 2$, then the runtime is $\softO(m^{\frac{4}{3}})$.
Despite intense research, there is no known better algorithm for triangle finding in graphs than that from Theorem~\ref{thm:AYZ}, which motivates the following common hypothesis found e.g.~in~\cite{AbboudW14}.
\AP \begin{hypothesis}[\intro{Triangle Hypothesis}]
    There is no algorithm that, given an input graph $G$ with~$m$ edges, decides in time $\softO(m)$ if $G$ contains a triangle.
\end{hypothesis}
A common more concrete lower bound conjecture for triangle finding is $\Omega(m^{4/3})$, which, as we have seen, is tight in case $\omega = 2$.

For conjunctive query answering, we get the following result.

\begin{proposition}\label{prop:embed-triangle}
    Let $q$ be cyclic self-join free Boolean conjunctive query in which all relations have arity $2$. Then, assuming the \kl{Triangle Hypothesis}, there is no algorithm that, given a database $D$ of size $m$, decides~$q$ on $D$ in time $\softO(m)$.
\end{proposition}
\begin{proof}
    We reduce from triangle finding, so let $G= (V,E)$ be the input graph. We construct a database $D$ with active domain $V \cup \{d\}$ where $d$ is a dummy element not in $V$. Fix an induced cycle in the graph of the query $q$, which exists because $q$ is cyclic. Let $R_1(x_1, x_2)$, $R_2(x_3, x_4)$, $R_3(x_5, x_6)$ be different atoms on the cycle, which might overlap in their endpoints. We set all $R_1^D = R_2^D = R_3^D = E$ and all other relations for the atoms on the cycle to the equality relation on $V$. For all relations $R(x_j, y)$ where $x_j$ is on the cycle and $y$ is not, we set $R^{D}:= V \times \{d\}$ and proceed analogously for the atoms $R_i(y, x_j)$. Finally, for all atoms $R(y_1, y_2)$ where neither $y_1$ nor $y_2$ is on the cycle, we set the input relation to $\{d\}^2$.
    
    Clearly, $D\models q$ if and only if $G$ has a triangle, so the claim follows directly.
\end{proof}

\subsubsection{General Queries}

As we have seen in the last section, in the graph-like case, i.e.~the case of arity $2$, having a cycle is the only cause for hardness of Boolean queries. In the general case, we will now see another hard substructure.
\begin{example}[Loomis-Whitney Joins]
    The $k$-dimensional Loomis-Whitney query is the query in variables $X_k:=\{x_1, \ldots, x_k\}$ defined by 
    \begin{align*}
        q^{LW}_k \datarule \bigwedge_{X\subseteq X_k: |X| = k-1} R_X(X).
    \end{align*}
    Note that for $k>3$ no Loomis-Whitney query contains an induced cycle, since all variable sets of size at most $k-1$ are contained in the scope of one atom. Moreover, it is known that any Loomis-Whitney query can be answered in time $\softO(m^{1+ \frac{1}{k-1}})$~\cite{NgoPRR18} which for $k \ge 5$ is faster than the state of the art for triangles. 
\end{example}

We will next give evidence that the known algorithms for Loomis-Whitney queries might be optimal, building on hypotheses from fine-grained complexity. To this end, define a hyperclique in a $h$-uniform hypergraph $H= (V,E)$ be a set of vertices $V'\subseteq V$ such that every set $S\subseteq V'$ of size $h$ is contained in $E$.
\AP \begin{hypothesis}[\intro{Hyperclique Hypothesis}]\label{hyp:hyperclique}
    For no pair $k>h>2$ of integers, there is an $\epsilon > 0$ and an algorithm that, given a $h$-uniform hypergraph $H$ with $n$ vertices, decides in time $\softO(n^{k-\varepsilon})$ if $H$ contains a hyperclique of size $k$.
\end{hypothesis} 
Note that the condition $h>2$ is necessary since for triangles and more general $k$-cliques in graphs there are algorithms with runtime $O(n^{ck})$ for some $c<1$ based on matrix multiplication~\cite{NesetrilP85}, see also Section~\ref{sct:clique}. In contrast, despite intensive search, no such algorithms for hyperclique in $h$-regular hypergraphs are known for any $h>2$. In particular, it is known that there is no generalization of efficient matrix multiplication to higher order tensors that would allow adapting the algorithms for $k$-clique in a direct way~\cite{LincolnWW18}. Moreover, improving algorithms for hypercliques would give an improvement for Max-$k$-SAT~\cite{LincolnWW18}, a problem that, in contrast to Max-$2$-SAT and $k$-SAT, has so far resisted all tries to improve upon the trivial runtime $\softO(2^n)$ for $n$-variable CNF-formulas.

Since Loomis-Whitney queries can be used to find cliques of size $k$ in $(k-1)$-regular hypergraphs, we directly get the following lower bound that the algorithm from~\cite{NgoPRR18} is optimal.

\begin{theorem}
    Assuming the \kl{Hyperclique Hypothesis}, for no $k>4$ and no $\varepsilon >0$, there is an algorithm that, given an input databaseof size $m$, decides $q^{LW}_k$ in time $\softO(m^{1+\frac{1}{k-1}-\varepsilon})$.
\end{theorem}
\begin{proof}
    Assume by way of contradiction that there are $k>4$ and $\varepsilon$ such that there is an algorithm for $q^{LW}_k$ with the stated runtime. We will solve the hyperclique problem in $(k-1)$-regular hypergraphs with the help of $q^{LW}_k$. Let $H=(V,E)$ be a $(k-1)$-regular hypergraph. We define a relation $R$ that contains for all edges $e\in E$ all tuples we get by permutation of the elements in $e$. Let~$R$ be the relation for every $R_X$, then $q^{LW}_k$ is true on $D$ if and only if $H$ has a hyperclique of size~$k$, so $q^{LW}_k$ lets us find $(k-1)$-hypercliques, as desired.
    
    Let $n:= |V|$, then $R$ has size at most $n^{k-1}$. So overall, the runtime of the algorithm is at most $\softO\left(\left(n^{k-1}\right)^{1+ \frac{1}{k-1}-\varepsilon}\right) = \softO\left(n^{k - (k-1)\varepsilon}\right)$, which contradicts the \kl{Hyperclique Hypothesis}.
\end{proof}
It follows easily that queries that contain Loomis-Whitney queries as subqueries cannot be answered in linear time under reasonable assumptions. Interestingly, this yields a characterization of all non-linear time queries due to the following result by Brault-Baron~\cite{BraultBaron13}.

\begin{theorem}[\cite{BraultBaron13}]\label{thm:brault-baron}
    Let $H$ be a hypergraph that is not acyclic. Then there is a set $S$ such that the induced hypergraph $H[S]$ is a cycle or we can get a $(|S|-1)$-uniform hyperclique from $H[S]$ by deleting edges that are completely contained in other edges.
\end{theorem}

Brault-Baron inferred the following characterization of linear-time decidable Boolean queries.

\begin{theorem}\label{thm:booleanLower}
    Assume the \kl{Triangle Hypothesis} and the \kl{Hyperclique Hypothesis}. Then there is a linear-time evaluation algorithm for a self-join free Boolean conjunctive query $q$ if and only if $q$ is acyclic.
\end{theorem}
\begin{proof}[Proof sketch]
    For the upper bound, we use the Yannakakis algorithm of Theorem~\ref{thm:yannakakis}. For the lower bound, we use Theorem~\ref{thm:brault-baron} to embed triangle or hyperclique finding, similarly to Proposition~\ref{prop:embed-triangle}.
\end{proof}
\subsection{Counting}\label{sct:counting}

In this section, we consider the question of counting answers to conjunctive queries as a prototypical aggregation question. In particular, we will show for which queries we can expect algorithms that allow counting answers in linear time. Since any such algorithm allows in particular answering if there is any answer, we get from Theorem~\ref{thm:booleanLower} that for self-join free queries, we can restrict our attention to acyclic queries. In the case of join queries, i.e.~when all variables are output variables, this already characterizes the linear-time solvable queries, since there is a generalization of the Yannakakis algorithm, see e.g.~\cite[Lemma 19]{BraultBaron13}.

\begin{theorem}\label{thm:dalmauJ}
    Assume the \kl{Triangle Hypothesis} and the \kl{Hyperclique Hypothesis}. Then there is~a~$\softO(m)$ algorithm that counts the answers to a join query $q$ on instances of size $m$ if and only if $q$ is acyclic.
\end{theorem}
Note that in Theorem~\ref{thm:dalmauJ} we do not require self-join freeness, since self-joins can be dealt with in the lower bound with an interpolation argument first used in~\cite{DalmauJ04}.

The situation becomes more interesting for queries with projections. Pichler and Skritek~\cite{PichlerS13} were the first to show that at least in \emph{combined} complexity projections make counting hard. Following this, a line of work including~\cite{DurandM14,DurandM13,ChenM15,GrecoS14,DellRW19} lead to a good understanding of tractable classes of queries, both in parameterized and combined complexity. The first results for fixed queries are found in~\cite{DellRW19}. However, those only apply to runtimes that are at most cubic, so we present a variant of the ideas from~\cite{DellRW19} here that is taken from the note~\cite{Mengel21}.

We will need some more definitions: a dominating set $S$ of a graph $G=(V,E)$ is a set $S\subseteq V$ such that every vertex not in $S$ has a neighbor in $S$. The problem $k$-Dominating Set (short $k$-DS) is to decide, given a graph $G$, if $G$ has a dominating set of size at most $k$. 

Our aim will be to show the following lower bound for the star queries that have already been a crucial building block of the lower bounds in~\cite{DurandM14,DurandM13}.

\begin{lemma}\label{lem:stars}
    Let $k\in \mathbb{N}$ with $k\ge 2$. If there is an algorithm that counts answers to \[q^\star_k(x_1, \ldots, x_k) \datarule \bigwedge_{i\in [k]} R(x_i,z)\] in time $O(m^{k-\varepsilon})$ on databases with $m$ tuples for some $\varepsilon>0$, then there is a $k'\in \mathbb{N}, k'>3$ such that $k'$-DS can be decided in time $O(n^{k'-\epsilon})$ on graphs with $n$ vertices.
\end{lemma}

The existence of an algorithm for $k$-DS in Lemma~\ref{lem:stars} is very unlikely because it would contradict the \kl{Strong Exponential Time Hypothesis}, a very impactful hypothesis from fine-grained complexity~\cite{ImpagliazzoP01,ImpagliazzoPZ01}. 

\AP \begin{hypothesis}[\intro{Strong Exponential Time Hypothesis}]
    For every $\varepsilon > 0$, there is a $k\in \mathbb{N}$ such that $k$-SAT cannot be solved on instances with $n$ variables in time $\softO(2^{n(1-\varepsilon)})$.
\end{hypothesis}

SAT is certainly the most well-studied problem in complexity theory, and despite this effort and the existence of faster algorithms for restricted variants like $k$-SAT, no algorithm with a runtime $O(2^{n(1-\varepsilon)})$ is known. This makes the \kl{Strong Exponential Time Hypothesis} rather plausible.
The following connection is from~\cite{PatrascuW10}.

\begin{theorem}\label{thm:PatrascuW10}
    Assuming the \kl{Strong Exponential Time Hypothesis}, there is no constant $\epsilon'$ such that there is a constant $k$ and an algorithm for $k$-DS with runtime $O(n^{k-\varepsilon'}) $ on graphs with $n$ vertices.
\end{theorem}

We now return to our star queries.
\begin{proof}[Proof of Lemma~\ref{lem:stars}]
    Choose $k'$ as a fixed integer such that $k'>k^2/\epsilon$ and $k'$ is divisible by $k$. We will encode $k'$-DS into the query $q^\star_k$. To this end, let $G=(V, E)$ be an input for $k'$-DS. Set
    \[R:=\{(\vec{u}, v)\mid v\in V, \vec{u} = (u_1,\ldots, u_{k'/k})\in V^{k'/k}, \forall i \in [k'/k]: u_iv\notin E, u_i \ne v\}.\]
    Then any assignment $\vec{u}^1, \ldots, \vec{u}^k$ to $x_1, \ldots, x_k$ corresponds to a choice $S$ of at most $k'$ vertices in $G$. The set $S$ is a dominating set if and only if there is no vertex $v$ in $V$ that is not in~$S$ and has no neighbor in $S$. This is the case if and only if there is no $v\in V$ that assigned to $z$ makes $q^\star_k$ true. Thus the answers to $q^\star_k(x_1, \ldots, x_k)$ are exactly the assignments $\vec{u}^1, \ldots, \vec{u}^k$ that do \emph{not} correspond to dominating sets of size at most $k'$ in $G$. Since there are exactly $n^{k'}$ choices for the $\vec u_i$, any algorithm counting the answers to $q^\star_k(x_1, \ldots, x_k)$ directly yields an algorithm for $k'$-DS.
    
    We now analyze the runtime of the above algorithm. Note that the time for the construction of $R$ is negligible, so the runtime is essentially that of the counting algorithm for $q^\star_k(x_1, \ldots, x_k)$. First observe that the relation $R$ has at most $n^{\frac{k'}{k}+1}$ tuples. The exponent of the runtime of the counting algorithm is thus
    \begin{align*}
        \left(\frac{k'}{k}+1\right)(k-\epsilon) &= k'+k-\frac{k' \epsilon}{k} - \epsilon
         < k'+k - \frac{k^2 \epsilon}{\epsilon k} -\epsilon
         = k'-\epsilon
    \end{align*}
    where the inequality comes from the choice of $k'$ satisfying $k'> k^2/\epsilon$.
\end{proof}

\begin{corollary}\label{cor:star}
    If SAT has no algorithm with runtime $O(2^{n(1-\epsilon)})$ for any $\epsilon>0$, then there is no constant $k\in \mathbb{N}$, $k\ge 2$ and no $\epsilon' > 0$ such that there is an algorithm that counts answers to \[q^\star_k(x_1,\ldots, x_k) \datarule \bigwedge_{i\in [k]}R(x_i,z)\] in time $O(m^{k-\epsilon'})$ on databases with $m$ tuples. 
\end{corollary}

The queries $q^\star_k$ are important because they can be embedded into many other queries. They thus play a crucial role in understanding the parameterized complexity of counting for conjunctive queries, see e.g.~\cite{BraultBaron13,DurandM14,DellRW19}. One central notion that we need in our fine-grained setting is that of free-connex acyclic queries, a concept from~\cite{BaganDG07}. An acyclic conjunctive query with hypergraph $\mathcal{H}$ and free variables $S$ is called \emph{free-connex} if the hypergraph $\mathcal{H}\cup \{S\}$ that we get from $\mathcal{H}$ by adding $S$ as an edge is acyclic as well. Non Free-connex queries are hard in the following sense.

\begin{theorem}\label{thm:embedding}
    Let $q$ be a self-join free conjunctive query that is acyclic but not free-connex. Then, assuming 
    the \kl{Strong Exponential Time Hypothesis}
    there is no algorithm that counts the solutions of $q$ on a database with $m$ tuples in time $O(m^{2-\epsilon'})$ for any $\epsilon'$. 
\end{theorem}

The argument is a minimal adaption of one found in the full version of~\cite{BaganDG07} for enumeration\footnote{Unfortunately, this full version has never been published in a journal, but it can be found at \url{https://webusers.imj-prg.fr/~arnaud.durand/papers/BDGlongversion.pdf}. The techniques can also be found in~\cite[Chapter 2.7]{Bagan09}.}. The idea is to embed the query~$q^\star_2$ and then use Corollary~\ref{cor:star}. For the details, we refer to~\cite{Mengel21}.

From Theorem~\ref{thm:embedding}, we get the following dichotomy theorem for linear time counting.
\begin{theorem}\label{thm:dichoCounting}
    Assume the \kl{Strong Exponential Time Hypothesis}, the \kl{Triangle Hypothesis} and the \kl{Hyperclique Hypothesis}. Then there is~a~$\softO(m)$ algorithm that counts the answers to a self-join free conjunctive  query $q$ on instances of size $m$ if and only if $q$ is free-connex acyclic.
\end{theorem}
\begin{proof}[Proof (sketch)]
The upper bound for Theorem~\ref{thm:dichoCounting} comes from the fact that for free-connex acyclic queries, one can eliminate all projected variables efficiently (see the proof and the discussion in~\cite[Section 4.1]{BerkholzGS20}), and then use the counting variant of the Yannakakis algorithm as in Theorem~\ref{thm:dalmauJ}. 

For the lower bound, if $q$ is cyclic, then by Theorem~\ref{thm:booleanLower} we get a lower bound using the fact that a counting algorithm also lets us decide the Boolean query we get from $q$ by projecting out all variables. If $q$ is cyclic by not free-connex, the lower bound follows from Theorem~\ref{thm:embedding}.
\end{proof}

\subsection{Enumeration}

We next turn to constant delay enumeration, which is the following problem: given a query and a database, first, in the \emph{preprocessing phase}, compute a data structure that is then used in the \emph{enumeration phase} to print out the answers to the query one after the other without repetition. There are two different time bounds in enumeration algorithms: first the bound on the preprocessing phase and then the \emph{delay} which is the maximal allowed time between printing out two answers. The most studied version in database theory is \emph{constant delay enumeration} where we want the delay between answers to be independent of the database and depend only on the query, which we consider fixed; this model has first been introduced in the very influential paper~\cite{BaganDG07} and since then studied extensively throughout database theory. Somewhat surprisingly, there are many query evaluation problems which allow constant delay enumeration, see e.g.~the survey~\cite{BerkholzGS20}.

Enumeration for conjunctive queries was already studied in~\cite{BaganDG07}, the first paper to study constant delay enumeration, which might also have been the first in database theory to use fine-grained complexity even before that field was well-established and without making an explicit connection. For the case of join queries, similarly to counting, from a lower bound perspective nothing too interesting happens: acyclic join queries allow constant delay enumeration after linear preprocessing~\cite{Bagan09}\footnote{The proof in~\cite{Bagan09} uses a translation of relational conjunctive queries to a fragment of functional first-order logic and then works in that setting which makes the argument somewhat difficult to follow for the reader not used to that setting. There are by now algorithms working directly on the relational database~\cite[Lemma~19]{BraultBaron13} or using factorized representations~\cite{OlteanuZ15}. For an accessible presentation see the tutorial~\cite{BerkholzGS20}.},~and since an enumeration algorithm can in particular be used to decide if there are any answers, we get from Theorem~\ref{thm:booleanLower} that acyclicity is the only property that yields constant delay after linear preprocessing (assuming the \kl{Triangle Hypothesis} and the \kl{Hyperclique Hypothesis}, of course).

\begin{theorem}\label{thm:enumlower}
    Assume the \kl{Triangle Hypothesis} and the \kl{Hyperclique Hypothesis}. Then there is~an enumeration algorithm for a self-join free join query $q$ with preprocessing time~$\softO(m)$ and constant delay on instances of size $m$ if and only if $q$ is acyclic.
\end{theorem}
We remark that, for the lower bound, since we are only interested in the first enumerated solution, we could instead go up to a linear delay bound and get the same statement.

The enumeration complexity of join queries with self-joins appears to be far harder to understand than that of self-join free queries. It was first observed in~\cite{BerkholzGS20} that the are cyclic join queries with self-joins whose answers can be enumerated with constant delay after linear preprocessing. This line of work was extended in~\cite{CarmeliS23}, giving many more examples where this happens and some more systematic understanding. Overall, the situation seems to be very subtle in the sense that apparently minor variations in the queries can have a big impact on their complexity. The complexity in that setting is thus not very well understood, and even for some concrete example queries, the enumeration complexity is unknown. We will thus restrict to self-join free queries in the remainder of this section.

Self-join free queries get more interesting in the case with projections. We will consider a self-join free variant of the star queries we saw already in Section~\ref{sct:counting}.
\begin{align*}
    \bar q^\star_k(x_1, \ldots, x_k) \datarule \bigwedge_{i\in [k]} R_i(x_i,z)
\end{align*}

\begin{theorem}\cite{BaganDG07}\label{thm:bagan}
    Assuming the \kl{Sparse Boolean Matrix Multiplication Hypothesis}, there is no enumeration algorithm for $\bar q^\star_2$ on instances of size $m$ with preprocessing $\softO(m)$ and delay $\softO(1)$.
\end{theorem}
In~\cite{BaganDG07} the complexity assumption is that (non-sparse) Boolean matrix multiplication cannot be done in quadratic time, but the proof actually yields the stronger statement we give here. 

\begin{proof}[Proof of Theorem~\ref{thm:bagan}]
    We will show how to use $\bar q^\star_2$ for Boolean matrix multiplication. So let~$A$ and $B$ be matrices given as lists of indices $(i,j)$ of their at most $m$ non-zero entries. We define a database $D$ for $\bar q^\star_2$ by setting $R_1^D := A$ and $R_2^D := \{(j,i)\mid (i,j)\in B\}$ (so the transpose of the matrix). Then $\bar q^\star_2(D)$ contains exactly the non-zero entries of the Boolean matrix product $AB$.
    
    Let $m'$ be the maximal number of non-zero entries in $A$,$B$ and $AB$. Assume by way of contradiction that there is an enumeration algorithm with preprocessing $O(m)$ and delay $\softO(1)$. Then computing all non-zero entries in $AB$ using this algorithm as above takes time $\softO(m)$ for the preprocessing and $\softO(m')$ to enumerate the up to $m'$ non-zero entries of $AB$. So the overall runtime is $\softO(m')$ which contradicts the \kl{Sparse Boolean Matrix Multiplication Hypothesis}.
\end{proof}

Now with the exact same embedding as mentioned for Theorem~\ref{thm:embedding}, we get the following.

\begin{theorem}\label{thm:embeddingEnum}
    Let $q$ be a self-join free conjunctive query that is acyclic but not free-connex. Then, assuming the \kl{Sparse Boolean Matrix Multiplication Hypothesis}, there is no enumeration algorithm for~$q$ on a database with $m$ tuples with preprocessing time $\softO(m)$ and delay~$\softO(1)$. 
\end{theorem}

Since free-connex acyclic queries have algorithms with preprocessing $\softO(m)$ and delay $\softO(1)$~\cite{BaganDG07}, this directly yields the following characterization.

\begin{theorem}\label{thm:EnumLower}
    Assume the \kl{Triangle Hypothesis}, the \kl{Hyperclique Hypothesis} and the \kl{Sparse Boolean Matrix Multiplication Hypothesis}. Then there is an enumeration algorithm for a conjunctive query $q$ with preprocessing time $\softO(m)$ and delay~$\softO(1)$ if and only if $q$ is free-connex acyclic.
\end{theorem}

There is much more to say about constant delay enumeration then we can do here, but since much of it has been surveyed elsewhere, we refer the reader to~\cite{Segoufin14,Segoufin15,BerkholzGS20}. 

\subsection{Direct Access}\label{sct:directaccesslinear}
 A restricted version of enumeration is \emph{direct access} where the aim is to simulate an array containing the whole query result without necessarily materializing it. After preprocessing the input database for a restricted time, in the \emph{access phase}, given an integer $i$, the algorithm must return the $i$th element of the (simulated) array containing the query result; if there are less than $i$ answers, the algorithm must return an error. The time required to return an answer at the requested position~$i$ is called the \emph{access time} of the algorithm.
 
 Clearly, direct access algorithms easily allow enumeration of the query result: after preprocessing, simply request the outputs one after another in order by starting with $i=1$ and incrementing until there is an error because no answers are left. Since it is known that free-connex acyclic queries allow direct access with logarithmic access time and linear preprocessing~\cite{BraultBaron13,CarmeliTGKR23}, we directly get the following characterization.

\begin{theorem}\label{thm:DAwithoutOrder}
    Assume the \kl{Triangle Hypothesis}, the \kl{Hyperclique Hypothesis} and the \kl{Sparse Boolean Matrix Multiplication Hypothesis}. Then there is a direct access algorithm for a self-join free conjunctive query $q$ with preprocessing time $\softO(m)$ and delay~$\softO(1)$ if and only if $q$ is free-connex acyclic.
\end{theorem}

Theorem~\ref{thm:DAwithoutOrder} does not specify the order in which the answers are stored in the simulated array. However, it is natural to assume an order on the answers, which, as we will see, has great impact on the the complexity of the problem.  Many different ways of defining orders on query have been studied see e.g.~\cite{DeepHK22,TziavelisAGRY20,AmarilliBCM24}; we will only consider two types of orders here: lexicographic orders and orders given by sums of weights of domain values. We will discuss them in individual subsections below.

\subsubsection{Lexicographic Orders}

One natural way of ordering query answers is by lexicographic orders. To do so, we assume that there is an order $\preccurlyeq$ on the active domain on the database. Then for every order $\preceq$ of the variables of the query, we get an induced order on tuples: given distinct tuples $a$ and $b$, we let $x$ be the first variable in $\preceq$ on which $a$ and $b$ differ. Then we say $a \le b$ if the entry on position $x$ in $a$ is before that of $b$ w.r.t.~$\preccurlyeq$. Note that the order on the tuples is different for different variable orders.

\begin{example}
    Let the active domain be $\{0,1\}$ with $0\preccurlyeq 1$. Consider a query $q(x,y)$ and the query result $\{(0,0), (0,1), (1,0), (1,1)\}$. If the variable order has $x\preceq y$, then the answers are sorted as $(0,0)\le  (0,1) \le (1,0) \le (1,1)$. If we have that $y\preceq x$, then $(0,0)\le  (1,0) \le (0,1) \le (1,1)$.
\end{example}

It will be convenient to consider the \emph{testing problem} for conjunctive queries which is the following for any fixed conjunctive query $q$: given a database $D$, after preprocessing $D$ for a restricted time, the algorithm has to answer, given a tuple $a$, if $a\in q(D)$. Again, we separate preprocessing time and query time. Note that testing for join queries is always easy, since we only have to check if $a$ is consistent with the relations of all atoms of $q$. On the other end of the spectrum, for Boolean conjunctive queries, the only possible input $a$ is the empty assignment, so the testing problem is equivalent to deciding Boolean queries. We will be interested in the middle ground where some but not all variables are projected.

The following is useful for direct access lower bounds, see e.g.~\cite{BringmannCM22}. We denote by $\|D\|$ the size of a database $D$.

\begin{lemma}\label{lem:DAtoTesting}
    Let $q(x_1, \ldots, x_r)\datarule R_1(X_1)\land \ldots \land R_\ell(X_\ell)$ be a join query and let $\pi$ be a variable order for $q$. Let $X'$ be a variable set consisting of a prefix of $\pi$. Then, if there is a direct-access algorithm for $q(x_1, \ldots, x_k)$ respecting the lexicographic order induced by $\pi$ with preprocessing time $p(m)$ and access time $a(m)$ for databases of size~$m$, then there is a testing algorithm for the projected query $q(X')$ with the same body with preprocessing time $p(m)$ and access time $\softO(\log(M(m)) a(m))$ where $M(m)$ is an upper bound with $M(m) \ge \max_{D:\|D\| \le m}(|q(D)| \mid \|D\| \le m)$.
\end{lemma}
\begin{proof}
    Since $X'$ contains the variables of a prefix, for every assignment $a$ to $X'$ all tuples consistent with $a$ in $q(D)$ are in a contiguous block of the array simulated by the direct access algorithm. Also, the order of those blocks is given by the lexicographic order on $X'$, so we can use binary search to check if a non-empty block for a given assignment to $X'$ exists. This directly gives the claimed runtime bound.
\end{proof}

Since we are interested in fixed conjunctive queries, we can bound $\log(M(m)) \le O(\log(m))$, so we only lose a logarithmic factor when going from direct access to testing.

It turns out that understanding the query~$q^\star_2$ is crucial to the understanding of lexicographic direct access due to the following lower bound. 

\begin{lemma}\label{lem:testingLB}
    Assuming the \kl{Triangle Hypothesis}, there is no testing algorithm for $q^\star_2$ with preprocessing time $\softO(m)$ and testing time $\softO(1)$ on databases of size $m$.
\end{lemma}
\begin{proof}
    By way of contradiction, assume that an algorithm with preprocessing time $\softO(m)$ and testing time $\softO(1)$ exists. We use it to detect triangles in graphs as follows: given a graph $G=(V,E)$, construct a database $D$ for $q^\star_2$ by setting $R^D = E$. Apply the preprocessing on $D$. Afterwards, for every $(a,b)\in E$ test if $(a,b)\in q^\star_2(D)$. If there is such an edge, return true, otherwise false.
    
    The algorithm is correct, since there is an edge $(a,b)\in E$ with $(a,b)\in q^\star_2(D)$ if and only if there is a triangle $(a,b,c)$ in $G$. To bound the runtime, first observe that the preprocessing takes time $\softO(\|D\|)= \softO(|E|)$. Then we make $|E|$ queries, each in time $\softO(1)$. So the overall runtime of the algorithm is $\softO(|E|)$, which contradicts the \kl{Triangle Hypothesis}.
\end{proof}

Embedding $q^\star_2$, we get a version of Theorem~\ref{thm:DAwithoutOrder} without the \kl{Sparse Boolean Matrix Multiplication Hypothesis}, if we insist on lexicographic orders.
\begin{corollary}\label{thm:DAlex1}
    Assume the \kl{Triangle Hypothesis} and the \kl{Hyperclique Hypothesis}. Then there is a direct access algorithm for a conjunctive query $q$ for a lexicographic order with preprocessing time $\softO(m)$ and delay~$\softO(1)$ if and only if $q$ is free-connex acyclic.
\end{corollary}

In the algorithm of Theorem~\ref{thm:DAwithoutOrder} and Corollary~\ref{thm:DAlex1}, we cannot choose the lexicographic order of the variables but it depends on the query. So a natural question is: can we choose any lexicographic order and still get linear preprocessing for direct access? This question was answered negatively in~\cite{CarmeliTGKR23}. A crucial role is played by the following variant of $q^\star_k$:
\begin{align*}
    \hat q^\star_k(x_1, \ldots, x_k, z) \datarule \bigwedge_{i\in [k]} R(x_i,z).
\end{align*}
Note that the only difference between $\hat q^\star_k$ and $q^\star_k$ is that $z$ is not projected in the former. Plugging together Lemma~\ref{lem:DAtoTesting} and Lemma~\ref{lem:testingLB}, we get the following.

\begin{lemma}\label{lem:trio}
    Assuming the \kl{Triangle Hypothesis}, direct access for $\hat q^\star_2$ and the variable order $x_1>x_2>z$ cannot be done on databases of size $m$ with preprocessing time $\softO(m)$ and access time $\softO(1)$.
\end{lemma}
\begin{proof}
    By Lemma~\ref{lem:DAtoTesting}, direct access for $\hat q^\star_2$ and the order $x_1>x_2>z$ allows testing for $q^\star_2$ which by Lemma~\ref{lem:testingLB} requires more than quasi-linear preprocessing.
\end{proof}

Lemma~\ref{lem:trio} motivates the following definition: let $q$ be a join query and $\preceq$ an order of its variables. Three variables $y_1, y_2, y_3$ are called a \emph{disruptive trio} if $y_1 \preceq y_3$, $y_2 \preceq y_3$, the pairs $y_1, y_3$ and $y_2, y_3$ each appear in the scope of a common atom of $q$ and $y_1, y_2$ do not appear in the scope of a common atom.

It is not hard to see that if a self-join free join query has a disruptive trio with respect to an order $\preceq$, then one can embed $\hat q^\star_2$ into it: define all relations such that all variables not in the disruptive trio are fixed to constants, and on the trio simulate the relations $R_1$ and $R_2$. It can be shown that such an embedding is actually also possible if the query contains self-joins~\cite{BringmannCM22}. As a consequence, disruptive trios make direct access hard for join queries. More work is required to show that if there is no disruptive trio, then the query allows efficient direct access after linear preprocessing, which is shown in~\cite{CarmeliTGKR23}. Together, this gives the following characterization.

\begin{theorem}\label{thm:lexlinear}
    Assume the \kl{Triangle Hypothesis} and the \kl{Hyperclique Hypothesis}. Then there is a direct access algorithm for a join query $q$ for a lexicographic order induced by an order $\preceq$ on databases of size $m$ with preprocessing time $\softO(m)$ and access time~$\softO(1)$ if and only if $q$ is acyclic and has no disruptive trio with respect to $\preceq$.
\end{theorem}
We remark that~\cite{CarmeliTGKR23} in fact shows a version of Theorem~\ref{thm:lexlinear} that has the \kl{Sparse Boolean Matrix Multiplication Hypothesis} as an additional assumption. This is because the authors there use Theorem~\ref{thm:embeddingEnum} in the lower bound which makes use of this assumption. By using the more direct approach of Lemma~\ref{lem:testingLB}, we could avoid introducing it for Theorem~\ref{thm:lexlinear}.

\subsubsection{Sum Orders}

We will now consider another way of ordering the tuples in a query result, this time by the \emph{sum order} which is defined as follows: assign a weight $w(a)$ to every element of the domain of the database, then the weight $w(a_1, \ldots, a_r)$ of a tuple $(a_1, \ldots, a_r)$ is the sum $\sum_{i\in [r]} w(a_i)$ of weights of its entries. The query result is then ordered by the weight defined this way. This type of orders is more expressive than lexicographic orders and has for example been considered in~\cite{DeepHK22,TziavelisAGRY20,AmarilliBCM24}. We will here present work from~\cite{CarmeliTGKR23}, however, in contrast to there, to simplify the presentation, we will only consider join queries, so queries without projections. We will see that the class of queries that allows linear time preprocessing direct access is extremely restricted in this setting.

We will use an additional hypothesis. The \emph{3SUM problem} is, given as input three lists $A$, $B$, $C$ of length $n$ consisting of integers in $\{-n^4,\ldots, n^4\}$, to decide if there are $a\in A$, $b\in B$ and $c\in C$ such that $a+b =c$. The 3SUM problem has an easy algorithm with complexity $\softO(n^2)$: compute the set $S:= \{a+b\mid a\in A, b\in B\}$ and sort it. Then going over the $S$ and $C$ in increasing order, check if there is a common element in both lists. While this runtime can be slightly improved, see e.g.~the discussion in~\cite{VassilevskaWilliams2018}, there is no known algorithm with runtime $\softO(n^{2-\varepsilon})$ for any constant $\varepsilon > 0$.

\AP \begin{hypothesis}[\intro{3SUM Hypothesis}]
    There is no algorithm for the 3SUM problem with runtime $\softO(n^{2-\varepsilon})$ for any $\varepsilon>0$.
\end{hypothesis}

The \kl{3SUM Hypothesis} was first introduced in~\cite{GajentaanO95} where it was connected to the complexity of many problems from computational geometry. Since then it has become one of the most used and important hypotheses in fine-grained complexity, see again~\cite{VassilevskaWilliams2018} for a discussion.

The following connection of 3SUM to direct access is from~\cite{CarmeliTGKR23}, in slightly different formulation.

\begin{lemma}\label{lem:sum}
    Assume the \kl{3SUM Hypothesis}. Let $q$ be a self-join free join query such that there are two variables $x$, $y$ such that no atom of $q$ contains both $x$ and $y$. Then there is no $\varepsilon> 0$ such that there is a direct access algorithm for $q$ on databases of size $m$ with sum order, preprocessing time $\softO(m^{2-\varepsilon})$ and access time $\softO(m^{1-\varepsilon})$.
\end{lemma}
\begin{proof}
    We reduce from 3SUM, so let $A$, $B$, $C$ be the input lists. We construct a database $D$ with domain $A\cup B\cup \{0\}$. In all relations, the variables outside of $x$ and $y$ can only take the value $0$, while $x$ can take all values in $A$ and $y$ all values in $B$. Since $x$ and $y$ do not appear in the scope of a common atom, the resulting database has size $O(n)$. As weight function $w$, we simply set $w(d) = d$ for all elements in the domain.
    
    By construction, $q(D)$ contains a tuple of weight $c$ if and only if there are $a\in B, b\in B$ such that $a+b=c$. We solve 3SUM as follows: first, run the preprocessing on $D$, then iterate over all elements $c\in C$ and check if $q(D)$ contains a tuple of weight $c$. To do the latter, use binary search with the help of $O(\log(n))$ access queries. Assuming preprocessing time $\softO(m^{2-\varepsilon})$ and access time $\softO(m^{1-\varepsilon})$, the overall runtime is $O(n^{2-\varepsilon})$ which breaks the \kl{3SUM Hypothesis}.
\end{proof}

One can show that the only acyclic queries that are not hard by Lemma~\ref{lem:sum} are those containing an atom which has all variables in its scope. This is because for acyclic hypergraphs the size of the smallest edge cover and that of the biggest independent set are equal~\cite[Lemma~19]{DurandM14}, so if the variables are not all in the scope of one atom, then the criterion of Lemma~\ref{lem:sum} applies. But for queries in which all variables are in one atom, we can easily materialize the query result and sort it in time~$\softO(m)$ and thus get efficient direct access. So we get the following characterization.

\begin{theorem}\label{thm:sumdicho}
    Assume the \kl{3SUM Hypothesis}. Then there is a direct access algorithm for an acyclic self-join free join query $q$ for sum order on databases of size $m$ with preprocessing time $\softO(m)$ and delay~$\softO(1)$ if and only if $q$ contains an atom which has all variables of $q$ in its scope.
\end{theorem}

\section{Beyond Linear Time}\label{sct:superlinear}

In this section, we will discuss to which extent the results seen before can be extended to showing superlinear lower bounds. Unfortunately, there are so far less powerful techniques for this, in particular for Boolean queries. There are mainly two problems we would have to overcome to show lower bounds similar to those of Section~\ref{sec:linear}: first, we would need an analogue of Theorem~\ref{thm:brault-baron} that shows that in ``hard'' instances we can always embed certain restricted and well-behaved sub-instances. Second, we need plausible hardness assumptions on these sub-instances.

Unfortunately, there seem to be no good answers to either of these problems that are as generally useful as in the linear time case. That said, there are techniques for both that allow us to still infer interesting results. We will sketch some of them in the next two sub-sections. After that, we will quickly mention results for counting and direct access that allow showing good lower bounds with slightly different techniques.

\subsection{Clique Variants}\label{sct:clique}

In many settings in complexity theory, the Clique problem has served as a starting point for hardness, see e.g.~\cite{Karp72} and~\cite[Chapter 13]{CyganFKLMPPS15}. Therefore, it is natural to base fine-grained lower bounds for query answering also on this problem.
To this end, it would be convenient to have a very strong hardness assumption for $k$-Clique, saying that $k$-Clique cannot be solved in time $O(n^{k-\varepsilon})$. Unfortunately, it has long been observed that this statement is false.

\begin{theorem}\cite{NesetrilP85}\label{thm:kclique}
    For every $k$, there is an algorithm with runtime $\softO(n^{\omega \lfloor\frac{k}{3}\rfloor+i})$ for $k$-clique where $i$ is the rest of $k$ after division by $3$.
\end{theorem}    
\begin{proof}
    We will only sketch the idea for $i=0$. We set $r:= \frac{k}{3}$. Let $G$ be an input graph. We construct~a new graph $G'$ whose vertices are the $r$-cliques of $G$. Two such cliques $C_1, C_2$ are connected by an edge if and only if they are disjoint and form a clique of size $2r$. There are at most $\binom{n}{r} = O(n^{\frac{k}{3}})$ potential vertices for $G'$ to consider and thus $O(n^{\frac{2k}{3}})$ potential edges in $G'$. It follows that we can construct $G'$ easily within the allowed time-bound.
    
    Now observe that the triangles in $G'$ are exactly the $k$-cliques in $G$. So we can use the matrix multiplication based algorithm for triangles as in the second part of the proof of Theorem~\ref{thm:AYZ} to decide in time $\softO(n^{\omega r})= \softO(n^{\frac{\omega k}{3}})$ if $G$ contains a $k$-clique.
\end{proof}

We remark that the exponents for $k$-clique that we get from Theorem~\ref{thm:kclique} by plugging in the best known bounds for $\omega$ can be slightly improved in the cases $i\in \{2,3\}$~\cite{EisenbrandG04}.

In the light of Theorem~\ref{thm:kclique}, $k$-Clique as such does not seem to be a good candidate for tight lower bound proofs, even though it is still useful in settings that have known matrix multiplication based algorithms, see e.g.~\cite{AbboudBW18}. There are two main ways of modifying the clique problem to make it more useful in our setting: either restricting the computational model or using harder variants of $k$-Clique. We will quickly discuss both of these approaches in the remainder of this section.

\subsubsection{Restricted Computation: Combinatorial Algorithms}

There is a line of work that considers only restricted, so-called \emph{combinatorial} algorithms for $k$-Clique and related problems. This notion is not formally defined, but mostly motivated by the fact that the only results that beat $\softO(n^k)$-time exhaustive search algorithms for $k$-Clique by a factor of $n^\varepsilon$ for any constant $\varepsilon > 0$ are all based on efficient matrix multiplication. This state of affairs is often considered as unsatisfying: algorithms using fast matrix multiplication are in a sense uninterpretable since intermediate values computed in them have no apparent interpretation in terms of the original problem. Also, these algorithms often do not generalize to problem variants, see the next section. Finally, it has long been argued that the asymptotically best known algorithms are impractical due to huge hidden constants, see e.g.~\cite{Gall12}, which makes all algorithms relying on them also impractical.

As a consequence of all these issues, many works consider combinatorial algorithms which, while the concept is not formally defined, are often informally understood to not use any algebraic cancellations and have reasonable constants, see e.g.~\cite[Section~1.2]{AbboudFS24} for a detailed discussion. A very popular conjecture for combinatorial algorithms is the following.

\AP \begin{hypothesis}[\intro{Combinatorial $k$-Clique Hypothesis}]
    Combinatorial algorithms cannot solve $k$-Clique in time $\softO(n^{k-\varepsilon})$ on graphs with $n$ vertices for any $\varepsilon>0$ and $k\ge 3$.
\end{hypothesis}

There are analogous hypotheses for combinatorial algorithms for other problems such as Boolean matrix multiplication and triangle finding, see~\cite{WilliamsW18}. Moreover, the best known combinatorial algorithms are faster than the naive approach and have runtime~$o(n^k)$; the current record is $O(n^k (\log(n))^{-(k+1)}\log\log(n)^{k+3})$~\cite{AbboudFS24}.

The \kl{Combinatorial $k$-Clique Hypothesis} has been used as a starting point to rule out combinatorial algorithms, see again~\cite{AbboudFS24} for references. In database theory, it has been used to argue that known combinatorial algorithms for e.g.~cycle joins and Loomis-Whitney joins are optimal~\cite{FanKZ23}.

So what should we take away from this line of work? First, lower bounds for combinatorial algorithms are certainly valuable in the sense that they give us insights on algorithms that one would actually like to implement in practice. Also, since clique is often a convenient problem to reduce from, lower bounds for combinatorial algorithms are often easier to get than lower bounds for general algorithms. That said, since ``combinatorial algorithms'' have no formal, generally agreed upon definition, the reader might feel somewhat uneasy to show these lower bounds. Mathematically speaking, what is it that is even shown? Also, there are situations in which there are (non-combinatorial) algorithms that are better than the lower bounds for combinatorial algorithms. We have seen this for $k$-clique, but the same is true for other problems, e.g.~parsing context free languages~\cite{Valiant75}. Similarly, there are matrix multiplication based algorithms for directed cycle detection that are better than any combinatorial algorithm~\cite{YusterZ04,DalirrooyfardVW21}. So in conclusion, while lower bounds based on the \kl{Combinatorial $k$-Clique Hypothesis} are useful, if possible one should at least complement them with lower bounds for general algorithms based on other credible assumptions. For example, we have seen this for Loomis-Whitney joins in Section~\ref{sct:boolean}.

\subsubsection{Harder Variants of $k$-Clique}\label{sct:harderclique}

Since, as we have seen, $k$-Clique is in a sense ``not hard enough'', several variants of $k$-Clique have been introduced that are presumably harder. We will sketch some here.

One variant of $k$-Clique that we have already seen in Section~\ref{sct:boolean} is the Hyperclique problem and the corresponding \kl{Hyperclique Hypothesis}. We have see there that this conjecture was useful for Loomis-Whitney joins and thus all cyclic queries. Unfortunately, hypergraphs are often harder to embed into other problems, so weighted variants of $k$-Clique on graphs have been studied.

To this end, we consider graphs $G=(V,E)$ that have an additional edge weight function $w:E\rightarrow \mathbb{Z}$. We consider two problems on such weighted graphs: Min-Weight-$k$-Clique is the problem, given an edge weighted graph $G$, to compute a $k$-clique $C$ such that the sum of the edge weights of $C$ is minimal. The Zero-$k$-Clique problem is, given an edge weighted graph $G$, to decide if $G$ contains a $k$-clique $C$ such that the sum of the edges weights in $C$ is $0$. It is widely conjectured that the addition of edge weights makes $k$-Clique harder and that both problems do not have algorithms that are much better than $\softO(n^k)$.

\AP \begin{hypothesis}[\intro{Min-Weight-$k$-Clique Hypothesis}]
    There is no algorithm that solves Min-Weight-$k$-Clique in time $\softO(n^{k-\varepsilon})$ on graphs with $n$ vertices for any $\varepsilon>0$ and $k\ge 3$.
\end{hypothesis}

\AP \begin{hypothesis}[\intro{Zero-$k$-Clique Hypothesis}]
    There is no algorithm that solves Zero-$k$-Clique in time $\softO(n^{k-\varepsilon})$ on graphs with $n$ vertices for any $\varepsilon>0$ and $k\ge 3$.
\end{hypothesis}

Observe that there are also node-weighted $k$-clique variants defined analogously to the problems above. However, for those problem there are matrix multiplication based algorithms that are faster than $\softO(n^k)$~\cite{CzumajL09,WilliamsW13}.

The \kl{Zero-$k$-Clique Hypothesis} sometimes also goes under the name \emph{Exact-Weight-$k$-Clique Hypothesis}. 
Both of the above conjectures have been used extensively in fine-grained complexity, for some recent pointers into the literature see e.g.~\cite{BringmannFHKKR24}. We will restrict ourselves here to applications in database theory.

The \kl{Min-Weight-$k$-Clique Hypothesis} is useful in the setting of aggregate queries over semirings~\cite{KhamisNR16}. To simplify the setting, let us only consider queries over the tropical $(\min, +)$-semiring, a setting which we quickly sketch here: consider a join query $q(X)\datarule R_1(X_1), \ldots, R_\ell(X_\ell)$. Let $D$ be a weighted database, i.e.~a database in which all tuples $t$ are assigned weights $w(t)$ from $\mathbb{R}\cup \{\infty\}$. We then extend this weight to query answers in $q(D)$ by making for $a\in q(D)$ the definition $w(a) := w(\pi_{X_1}(a)) + \ldots + w(\pi_{X_\ell}(a))$ where $\pi_{X_i}$ denotes projection onto the set $X_i$. The answer to the aggregate query is then defined as $\min_{a\in q(D)}(w(a))$.

If we consider the join query $q_k(x_1, \ldots, x_k) \datarule \bigwedge_{i,j\in [k], i\ne j} E(x_i, x_j)$, we get a reformulation of the Min-Weight-$k$-Clique problem, so hardness follows directly from the \kl{Min-Weight-$k$-Clique Hypothesis}. More interestingly, this can be combined with the so-called clique embeddings which we will discuss in Section~\ref{sct:embeddings}.

For applications of the \kl{Zero-$k$-Clique Hypothesis}, see Section~\ref{sct:superlinearDA}

\subsection{Clique Embeddings}\label{sct:embeddings}

Having seen several hardness assumptions for the clique problem, we will in this section sketch on an example how \emph{clique embeddings} for queries are defined and can be used to apply the hardness assumptions to more general queries. The techniques sketched here are essentially taken from~\cite{FanKZ23} and build on ideas from~\cite{Marx13}.

The idea is that we will let the variables of a query $q$ choose vertices in a graph that potentially might belong to a clique and then construct a database to use the query to check if all necessary edges between vertices are present in the input graph. To do so, we have to decide which variable in the query is responsible for choosing which vertex. To this end, we choose an assignment $\psi$ of the vertices of a clique $K_\ell$ to vertices in the hypergraph $H=(V,E)$ of the query $q$ we want to analyze. This assignment $\psi$ has the following properties:
\begin{enumerate}
    \item\label{1} every vertex $x_i$ of $K_\ell$ can be mapped to several vertices of $H$, so $\psi$ maps to the subsets of $V$. To make sure that the choice of nodes in the variables is consistent, we require that $\psi(x_i)$ is connected in $H$ for every vertex of $x_i$ of $K_\ell$.
    \item\label{3} for every pair $x_i, x_j$ of $K_\ell$, we have that $\psi(x_i) \cap \psi(x_j) \neq \emptyset$ or there exists some $e \in E$ such that $\psi(x_i) \cap e \neq \emptyset$ and $\psi(x_i) \cap e \neq \emptyset$. 
\end{enumerate}

\begin{example}\label{ex:5cycleembedding}
    Consider the $5$-cycle query 
    \begin{align*}
    q^\circ_5 \datarule R_1(v_1, v_2), \ldots, R_4(v_4, v_5), R_5(v_5, v_1).
    \end{align*}
    We embed a $5$-clique on vertices $x_1, \ldots, x_5$ by setting
    \begin{align*}
    \psi(x_1) &:= \{v_1, v_2, v_3\},\\
    \psi(x_2) &:= \{v_2, v_3, v_4\},\\
    \psi(x_3) &:= \{v_3, v_4, v_5\},\\
    \psi(x_4) &:= \{v_4, v_5, v_1\},\\
    \psi(x_5) &:= \{v_5, v_1, v_2\}
    \end{align*}
    This mapping is visualized in Figure~\ref{fig:cycle} and has the properties we described above.
\end{example}

\begin{figure}
\begin{tikzpicture}[thick]
    \def\r{2} 
    

\node[draw] (v1) at (-90:\r) {$v_1 : x_1, x_2, x_3$};    
\node[draw] (v2) at (-18:\r) {$v_2 : x_2, x_3, x_4$};    
\node[draw] (v3) at (54:\r) {$v_3 : x_3, x_4, x_5$};    
\node[draw] (v4) at (126:\r) {$v_4 : x_4, x_5, x_1$};    
\node[draw] (v5) at (198:\r) {$v_5 : x_5, x_1, x_2$};    
\draw (v1) -- (v2); 
\draw (v2) -- (v3); 
\draw (v3) -- (v4); 
\draw (v4) -- (v5); 
\draw (v5) -- (v1); 
\end{tikzpicture}
\caption{Visualization of the embedding in Example~\ref{ex:5cycleembedding}. In each node of the cycle, we first give the name of the node and then, separated by a colon, the list of vertices of the $5$-clique that are mapped to it.}\label{fig:cycle}
\end{figure}
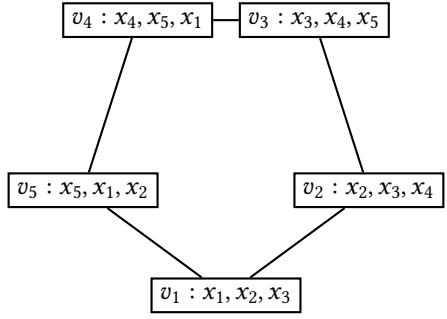

Having fixed $\psi$, we now show how to use $q$ to solve clique problems: given an input graph $G= (V_G,E_G)$, we construct a database $D$. The domain of a variable $v$ of $q$ is $V_G^{\psi^{-1}(v)}$, so intuitively,~$v$ choses a vertex of $G$ for every node$x_i$ in the clique $K_\ell$ with $v\in \psi(x_i)$. The relations of $D$ are constructed to satisfy two criteria: first, for every atom, if $v,u\in \psi(x_i)$ for some $x_i$ and $v,u$ are variables of the atom, then $v$ and $u$ must choose the same value for $x_i$. Because of property (\ref{1}) of $\psi$, it follows that all variables of $q$ make a consistent choice for all vertices $x_i$ of $K_\ell$. Second, whenever $x_i, x_j$ are mapped to variables $v, u \in e$ where $e$ is an edge of $G$ (here the case $v=u$ is possible), we make sure that the choice of vertices from $V_G$ is such that they are connected by an edge. Because of property (\ref{3}) of $\psi$, we get that the choices for all $x_i$ of $K_\ell$ are such that every two chosen vertices must be connected by an edge.

Defining the database like this, the answers to $q$ are essentially the $\ell$-cliques in $G$. So the idea is that, if we are trying to solve a problem that is hard for cliques, we can get a lower bound for the same problem on $q$. The exact lower bound depends on three quantities:
\begin{itemize}
    \item the hypothesis on the hardness of the problem on cliques that we make (see the previous section),
    \item the number of vertices from $K_\ell$ that are mapped into the same edge of $H$, since this determines the size of the corresponding relation, and
    \item the size $\ell$ of the clique that $\psi$ embeds.
\end{itemize}

\begin{example}\label{ex:embeddingpower}
    Consider again the embedding from Example~\ref{ex:5cycleembedding}. Let $|V_G| = n_G$. First, note that exactly $4$ variables are mapped to every edge, so the database has size $O(n_G^4)$. We have $\ell= 5$. Let us assume that we want to solve the min-weight $\ell$-clique problem which, assuming the \kl{Min-Weight-$k$-Clique Hypothesis} takes time $n_G^5$. Now assume that we can solve the aggregation problem for $q^\circ_5$ over the tropical semi-ring (see Section~\ref{sct:harderclique})in time $\softO(m^{\frac{5}{4}- \varepsilon})$. Then, by the encoding, we can solve it in $5$-cliques in time~$\softO(m^{\frac{5}{4}})=\softO(n^{4 \cdot \frac{5}{4}- \varepsilon}) =  \softO(n^{5-\varepsilon})$, which contradicts our assumption. So aggregation on the $5$-cycle query cannot be done in time $\softO(m^{\frac{5}{4}-\varepsilon})$ and we have reduced our hardness hypothesis for cliques to the $5$-cycle with the help of the embedding.
\end{example}


The techniques that we have only sketched here can be developed into a measure for queries called \emph{clique embedding power} which allows showing lower bounds similar to that in Example~\ref{ex:embeddingpower} systematically, see~\cite{FanKZ23} for details. For example, this approach allows showing lower bounds for aggregation over the tropical semiring for all cycle queries and Loomis-Whitney joins. Moreover, it can also be used to show lower bounds for the \emph{submodular width} of queries, an important hypergraph width measure that is known to yield algorithms that are optimal in parameterized complexity~\cite{Marx13,Khamis0S17}.

\subsection{Direct Access}\label{sct:superlinearDA}

The query answering task that is best understood out of those that we discuss in this paper, is certainly direct access with lexicographic orders. As we have seen in Section~\ref{sct:directaccesslinear}, the linear preprocessing case was solved in~\cite{CarmeliTGKR23}. In~\cite{BringmannCM22} this was generalized by introducing a parameter called \emph{incompatibility number} that is then shown to be the exact exponent of an optimal direct access algorithm. For details, we refer the reader to~\cite{BringmannCM22}.

Conceptually, the proof in~\cite{BringmannCM22} consists of two steps: in one step, one has to show tight lower bounds for one specific class of queries; in a second step, one has to show how to embed these queries. Concretely, the following generalization of Lemma~\ref{lem:testingLB} is shown for the queries~$\hat q_k^\star$.

\begin{lemma}\label{lem:testingLB2}
    Assuming the \kl{Zero-$k$-Clique Hypothesis}, there is no testing algorithm for $\hat q^\star_k$ with preprocessing time $\softO(m^{k-\varepsilon})$ and testing time $\softO(1)$ on databases of size $m$ for any $k\ge 2$, $\varepsilon>0$. 
\end{lemma}
Unfortunately, the proof of Lemma~\ref{lem:testingLB2} is far more involved than that of Lemma~\ref{lem:testingLB} and is thus beyond the scope of this paper. 
Very similar techniques are used in~\cite{BringmannCM22} to show a variant of Theorem~\ref{thm:enumlower} that does not depend on the Triangle and Hyperclique Hypotheses.

\begin{theorem}
        Assume the \kl{Zero-$k$-Clique Hypothesis}. Then there is~an enumeration algorithm for a self-join free join query $q$ with preprocessing time~$\softO(m)$ and constant delay on instances of size $m$ if and only if $q$ is acyclic.
\end{theorem}

This gives even more confidence that the characterization of Theorem~\ref{thm:enumlower} is likely correct.

\subsection{Counting for Acyclic Queries}

For counting for acyclic conjunctive queries, one can get a better understanding of the superlinear time cases. The idea is that one can can define a measure called \emph{quantified star size} that generalizes the property of being free-connex. This notion was first introduced in~\cite{DurandM14} for acyclic queries and then developed further in~\cite{DurandM13,ChenM15,DellRW19}. We will not give the slightly involved definition here; intuitively, it measures the size of the biggest query $q_k^\star$ from Section~\ref{sct:counting} that can be embedded into a query. Using Lemma~\ref{lem:stars}, we then get a lower bound.

\begin{theorem}
    Let $q$ be a self-join free conjunctive query of quantified star size
    $k$. Then, assuming that SAT has no algorithm with runtime $2^{n(1-\delta)}$ for any
    $\delta > 0$, there is no algorithm that counts the solutions of $q$ on a database with $m$
    tuples in time $m^{k-\varepsilon'}$ for any $\varepsilon'> 0$.
\end{theorem}

See~\cite{Mengel21} for more details.

For the non-acyclic case, tight lower bounds are mostly not known for the same reasons as for Boolean queries answering. However, the situation is even worse for counting than for decision since there is no apparent way for how to adapt the best known decision algorithms~\cite{Marx13,Khamis0S17} to counting, see~\cite{KhamisCMNNOS20}.

\section{Conclusion}

Let us conclude our short tour of recent lower bound techniques for conjunctive query answering. The aim was to convince the reader that fine-grained methods and hypotheses are a valuable new tool for database theorists. Let us stress that, while many of the reductions we presented here are relatively short and easy, this is certainly not the case for most results in fine-grained complexity in general. Quite to the contrary, many papers in the area are long and technical; the presentation here was merely restricted to rather simple techniques and ideas to showcase the area. However, the complicatedness of fine-grained complexity should not discourage us from approaching and integrating this area into our work: on the one hand, we have seen that relatively simple techniques can already lead to interesting results, so we should push them as far as we can. On the other hand, fine-grained complexity provides an exciting opportunity to learn interesting new techniques, adapt them and add them to our toolkit.

To keep this paper reasonably short, we have only focused conjunctive query answering in this survey and not mentioned other applications of fine-grained complexity in database theory. One important area that we neglected is query answering under updates for which many tight bounds are known, see e.g.~\cite{BerkholzKS17,BerkholzKS18,KaraNNOZ19,KaraNOZ20,KaraNOZ23}. There is also recent work that instead of only time studies also space bounds~\cite{ZhaoDK23}. Beyond this, queries with negation have been studied in~\cite{Brault-Baron12,BraultBaron13,ZhaoFOK24,CapelliCIS23}. There are also highly nontrivial results on answering UCQs\cite{BringmannC25,CarmeliK21}. Finally, lower bounds for regular path queries have been studied in~\cite{CaselS23}.

\bibliography{slides/pods}

\begin{thebibliography}{10}

\bibitem{AbboudBW18}
Amir Abboud, Arturs Backurs, and Virginia {Vassilevska Williams}.
\newblock {If the Current Clique Algorithms Are Optimal, so Is Valiant's
  Parser}.
\newblock {\em {SIAM} J. Comput.}, 47(6):2527--2555, 2018.
\newblock \href {https://doi.org/10.1137/16M1061771}
  {\path{doi:10.1137/16M1061771}}.

\bibitem{AbboudBFK24}
Amir Abboud, Karl Bringmann, Nick Fischer, and Marvin K{\"{u}}nnemann.
\newblock The time complexity of fully sparse matrix multiplication.
\newblock In David~P. Woodruff, editor, {\em Proceedings of the 2024 {ACM-SIAM}
  Symposium on Discrete Algorithms, {SODA} 2024, Alexandria, VA, USA, January
  7-10, 2024}, pages 4670--4703. {SIAM}, 2024.
\newblock \href {https://doi.org/10.1137/1.9781611977912.167}
  {\path{doi:10.1137/1.9781611977912.167}}.

\bibitem{AbboudFS24}
Amir Abboud, Nick Fischer, and Yarin Shechter.
\newblock Faster combinatorial k-clique algorithms.
\newblock In Jos{\'{e}}~A. Soto and Andreas Wiese, editors, {\em {LATIN} 2024:
  Theoretical Informatics - 16th Latin American Symposium, Puerto Varas, Chile,
  March 18-22, 2024, Proceedings, Part {I}}, volume 14578 of {\em Lecture Notes
  in Computer Science}, pages 193--206. Springer, 2024.
\newblock \href {https://doi.org/10.1007/978-3-031-55598-5\_13}
  {\path{doi:10.1007/978-3-031-55598-5\_13}}.

\bibitem{AbboudW14}
Amir Abboud and Virginia {Vassilevska Williams}.
\newblock Popular conjectures imply strong lower bounds for dynamic problems.
\newblock In {\em 55th {IEEE} Annual Symposium on Foundations of Computer
  Science, {FOCS} 2014, Philadelphia, PA, USA, October 18-21, 2014}, pages
  434--443. {IEEE} Computer Society, 2014.
\newblock \href {https://doi.org/10.1109/FOCS.2014.53}
  {\path{doi:10.1109/FOCS.2014.53}}.

\bibitem{AlmanDWXXZ25}
Josh Alman, Ran Duan, Virginia {Vassilevska Williams}, Yinzhan Xu, Zixuan Xu,
  and Renfei Zhou.
\newblock More asymmetry yields faster matrix multiplication.
\newblock In Yossi Azar and Debmalya Panigrahi, editors, {\em Proceedings of
  the 2025 Annual {ACM-SIAM} Symposium on Discrete Algorithms, {SODA} 2025, New
  Orleans, LA, USA, January 12-15, 2025}, pages 2005--2039. {SIAM}, 2025.
\newblock \href {https://doi.org/10.1137/1.9781611978322.63}
  {\path{doi:10.1137/1.9781611978322.63}}.

\bibitem{AlonYZ97}
Noga Alon, Raphael Yuster, and Uri Zwick.
\newblock Finding and counting given length cycles.
\newblock {\em Algorithmica}, 17(3):209--223, 1997.
\newblock \href {https://doi.org/10.1007/BF02523189}
  {\path{doi:10.1007/BF02523189}}.

\bibitem{AmarilliBCM24}
Antoine Amarilli, Pierre Bourhis, Florent Capelli, and Mika{\"{e}}l Monet.
\newblock Ranked enumeration for {MSO} on trees via knowledge compilation.
\newblock In Graham Cormode and Michael Shekelyan, editors, {\em 27th
  International Conference on Database Theory, {ICDT} 2024, March 25-28, 2024,
  Paestum, Italy}, volume 290 of {\em LIPIcs}, pages 25:1--25:18. Schloss
  Dagstuhl - Leibniz-Zentrum f{\"{u}}r Informatik, 2024.
\newblock URL: \url{https://doi.org/10.4230/LIPIcs.ICDT.2024.25}, \href
  {https://doi.org/10.4230/LIPICS.ICDT.2024.25}
  {\path{doi:10.4230/LIPICS.ICDT.2024.25}}.

\bibitem{AmossenP09}
Rasmus~Resen Amossen and Rasmus Pagh.
\newblock Faster join-projects and sparse matrix multiplications.
\newblock In Ronald Fagin, editor, {\em Database Theory - {ICDT} 2009, 12th
  International Conference, St. Petersburg, Russia, March 23-25, 2009,
  Proceedings}, volume 361 of {\em {ACM} International Conference Proceeding
  Series}, pages 121--126. {ACM}, 2009.
\newblock \href {https://doi.org/10.1145/1514894.1514909}
  {\path{doi:10.1145/1514894.1514909}}.

\bibitem{ArenasBLMP21}
Marcelo Arenas, Pablo Barcel\'o, Leonid Libkin, Wim Martens, and Andreas
  Pieris.
\newblock {\em Database Theory}.
\newblock Open access at \url{https://github.com/pdm-book/community}, 2022.
\newblock accessed March 2025, commit 9f403e4f8bb14eccca301eda2feafe90179b98df.

\bibitem{AtseriasGM13}
Albert Atserias, Martin Grohe, and D{\'{a}}niel Marx.
\newblock Size bounds and query plans for relational joins.
\newblock {\em {SIAM} J. Comput.}, 42(4):1737--1767, 2013.
\newblock \href {https://doi.org/10.1137/110859440}
  {\path{doi:10.1137/110859440}}.

\bibitem{Bagan09}
Guillaume Bagan.
\newblock {\em Algorithmes et complexit{\'{e}} des probl{\`{e}}mes
  d'{\'{e}}num{\'{e}}ration pour l'{\'{e}}valuation de requ{\^{e}}tes logiques.
  (Algorithms and complexity of enumeration problems for the evaluation of
  logical queries)}.
\newblock PhD thesis, University of Caen Normandy, France, 2009.
\newblock URL: \url{https://tel.archives-ouvertes.fr/tel-00424232}.

\bibitem{BaganDG07}
Guillaume Bagan, Arnaud Durand, and Etienne Grandjean.
\newblock On acyclic conjunctive queries and constant delay enumeration.
\newblock In Jacques Duparc and Thomas~A. Henzinger, editors, {\em Computer
  Science Logic, 21st International Workshop, {CSL} 2007, 16th Annual
  Conference of the EACSL, Lausanne, Switzerland, September 11-15, 2007,
  Proceedings}, volume 4646 of {\em Lecture Notes in Computer Science}, pages
  208--222. Springer, 2007.
\newblock \href {https://doi.org/10.1007/978-3-540-74915-8\_18}
  {\path{doi:10.1007/978-3-540-74915-8\_18}}.

\bibitem{BeeriFMY83}
Catriel Beeri, Ronald Fagin, David Maier, and Mihalis Yannakakis.
\newblock On the desirability of acyclic database schemes.
\newblock {\em J. {ACM}}, 30(3):479--513, 1983.
\newblock \href {https://doi.org/10.1145/2402.322389}
  {\path{doi:10.1145/2402.322389}}.

\bibitem{BerkholzGS20}
Christoph Berkholz, Fabian Gerhardt, and Nicole Schweikardt.
\newblock Constant delay enumeration for conjunctive queries: a tutorial.
\newblock {\em {ACM} {SIGLOG} News}, 7(1):4--33, 2020.
\newblock \href {https://doi.org/10.1145/3385634.3385636}
  {\path{doi:10.1145/3385634.3385636}}.

\bibitem{BerkholzKS17}
Christoph Berkholz, Jens Keppeler, and Nicole Schweikardt.
\newblock Answering conjunctive queries under updates.
\newblock In Emanuel Sallinger, Jan~Van den Bussche, and Floris Geerts,
  editors, {\em Proceedings of the 36th {ACM} {SIGMOD-SIGACT-SIGAI} Symposium
  on Principles of Database Systems, {PODS} 2017, Chicago, IL, USA, May 14-19,
  2017}, pages 303--318. {ACM}, 2017.
\newblock \href {https://doi.org/10.1145/3034786.3034789}
  {\path{doi:10.1145/3034786.3034789}}.

\bibitem{BerkholzKS18}
Christoph Berkholz, Jens Keppeler, and Nicole Schweikardt.
\newblock Answering ucqs under updates and in the presence of integrity
  constraints.
\newblock In Benny Kimelfeld and Yael Amsterdamer, editors, {\em 21st
  International Conference on Database Theory, {ICDT} 2018, March 26-29, 2018,
  Vienna, Austria}, volume~98 of {\em LIPIcs}, pages 8:1--8:19. Schloss
  Dagstuhl - Leibniz-Zentrum f{\"{u}}r Informatik, 2018.
\newblock URL: \url{https://doi.org/10.4230/LIPIcs.ICDT.2018.8}, \href
  {https://doi.org/10.4230/LIPICS.ICDT.2018.8}
  {\path{doi:10.4230/LIPICS.ICDT.2018.8}}.

\bibitem{Blaser03}
Markus Bl{\"{a}}ser.
\newblock On the complexity of the multiplication of matrices of small formats.
\newblock {\em J. Complex.}, 19(1):43--60, 2003.
\newblock \href {https://doi.org/10.1016/S0885-064X(02)00007-9}
  {\path{doi:10.1016/S0885-064X(02)00007-9}}.

\bibitem{Brault-Baron12}
Johann Brault{-}Baron.
\newblock A negative conjunctive query is easy if and only if it is
  beta-acyclic.
\newblock In Patrick C{\'{e}}gielski and Arnaud Durand, editors, {\em Computer
  Science Logic (CSL'12) - 26th International Workshop/21st Annual Conference
  of the EACSL, {CSL} 2012, September 3-6, 2012, Fontainebleau, France},
  volume~16 of {\em LIPIcs}, pages 137--151. Schloss Dagstuhl - Leibniz-Zentrum
  f{\"{u}}r Informatik, 2012.
\newblock URL: \url{https://doi.org/10.4230/LIPIcs.CSL.2012.137}, \href
  {https://doi.org/10.4230/LIPICS.CSL.2012.137}
  {\path{doi:10.4230/LIPICS.CSL.2012.137}}.

\bibitem{BraultBaron13}
Johann Brault{-}Baron.
\newblock {\em De la pertinence de l'{\'{e}}num{\'{e}}ration : complexit{\'{e}}
  en logiques propositionnelle et du premier ordre. (The relevance of the list:
  propositional logic and complexity of the first order)}.
\newblock PhD thesis, University of Caen Normandy, France, 2013.
\newblock URL: \url{https://tel.archives-ouvertes.fr/tel-01081392}.

\bibitem{Brault-Baron16}
Johann Brault{-}Baron.
\newblock Hypergraph acyclicity revisited.
\newblock {\em {ACM} Comput. Surv.}, 49(3):54:1--54:26, 2016.
\newblock \href {https://doi.org/10.1145/2983573} {\path{doi:10.1145/2983573}}.

\bibitem{BringmannC25}
Karl Bringmann and Nofar Carmeli.
\newblock Unbalanced triangle detection and enumeration hardness for unions of
  conjunctive queries.
\newblock {\em Log. Methods Comput. Sci.}, 21(1), 2025.
\newblock URL: \url{https://doi.org/10.46298/lmcs-21(1:29)2025}, \href
  {https://doi.org/10.46298/LMCS-21(1:29)2025}
  {\path{doi:10.46298/LMCS-21(1:29)2025}}.

\bibitem{BringmannCM22}
Karl Bringmann, Nofar Carmeli, and Stefan Mengel.
\newblock Tight fine-grained bounds for direct access on join queries.
\newblock In Leonid Libkin and Pablo Barcel{\'{o}}, editors, {\em {PODS} '22:
  International Conference on Management of Data, Philadelphia, PA, USA, June
  12 - 17, 2022}, pages 427--436. {ACM}, 2022.
\newblock \href {https://doi.org/10.1145/3517804.3526234}
  {\path{doi:10.1145/3517804.3526234}}.

\bibitem{BringmannFHKKR24}
Karl Bringmann, Nick Fischer, Ivor {van der Hoog}, Evangelos Kipouridis, Tomasz
  Kociumaka, and Eva Rotenberg.
\newblock Dynamic dynamic time warping.
\newblock In David~P. Woodruff, editor, {\em Proceedings of the 2024 {ACM-SIAM}
  Symposium on Discrete Algorithms, {SODA} 2024, Alexandria, VA, USA, January
  7-10, 2024}, pages 208--242. {SIAM}, 2024.
\newblock \href {https://doi.org/10.1137/1.9781611977912.10}
  {\path{doi:10.1137/1.9781611977912.10}}.

\bibitem{CapelliCIS23}
Florent Capelli, Nofar Carmeli, Oliver Irwin, and Sylvain Salvati.
\newblock Direct access for conjunctive queries with negation.
\newblock {\em CoRR}, abs/2310.15800, 2023.
\newblock URL: \url{https://doi.org/10.48550/arXiv.2310.15800}, \href
  {https://arxiv.org/abs/2310.15800} {\path{arXiv:2310.15800}}, \href
  {https://doi.org/10.48550/ARXIV.2310.15800}
  {\path{doi:10.48550/ARXIV.2310.15800}}.

\bibitem{CarmeliK21}
Nofar Carmeli and Markus Kr{\"{o}}ll.
\newblock On the enumeration complexity of unions of conjunctive queries.
\newblock {\em {ACM} Trans. Database Syst.}, 46(2):5:1--5:41, 2021.
\newblock \href {https://doi.org/10.1145/3450263} {\path{doi:10.1145/3450263}}.

\bibitem{CarmeliS23}
Nofar Carmeli and Luc Segoufin.
\newblock Conjunctive queries with self-joins, towards a fine-grained
  enumeration complexity analysis.
\newblock In Floris Geerts, Hung~Q. Ngo, and Stavros Sintos, editors, {\em
  Proceedings of the 42nd {ACM} {SIGMOD-SIGACT-SIGAI} Symposium on Principles
  of Database Systems, {PODS} 2023, Seattle, WA, USA, June 18-23, 2023}, pages
  277--289. {ACM}, 2023.
\newblock \href {https://doi.org/10.1145/3584372.3588667}
  {\path{doi:10.1145/3584372.3588667}}.

\bibitem{CarmeliTGKR23}
Nofar Carmeli, Nikolaos Tziavelis, Wolfgang Gatterbauer, Benny Kimelfeld, and
  Mirek Riedewald.
\newblock Tractable orders for direct access to ranked answers of conjunctive
  queries.
\newblock {\em {ACM} Trans. Database Syst.}, 48(1):1:1--1:45, 2023.
\newblock \href {https://doi.org/10.1145/3578517} {\path{doi:10.1145/3578517}}.

\bibitem{CaselS23}
Katrin Casel and Markus~L. Schmid.
\newblock Fine-grained complexity of regular path queries.
\newblock {\em Log. Methods Comput. Sci.}, 19(4), 2023.
\newblock URL: \url{https://doi.org/10.46298/lmcs-19(4:15)2023}, \href
  {https://doi.org/10.46298/LMCS-19(4:15)2023}
  {\path{doi:10.46298/LMCS-19(4:15)2023}}.

\bibitem{ChandraM77}
Ashok~K. Chandra and Philip~M. Merlin.
\newblock Optimal implementation of conjunctive queries in relational data
  bases.
\newblock In John~E. Hopcroft, Emily~P. Friedman, and Michael~A. Harrison,
  editors, {\em Proceedings of the 9th Annual {ACM} Symposium on Theory of
  Computing, May 4-6, 1977, Boulder, Colorado, {USA}}, pages 77--90. {ACM},
  1977.
\newblock \href {https://doi.org/10.1145/800105.803397}
  {\path{doi:10.1145/800105.803397}}.

\bibitem{ChenM15}
Hubie Chen and Stefan Mengel.
\newblock A trichotomy in the complexity of counting answers to conjunctive
  queries.
\newblock In Marcelo Arenas and Mart{\'{\i}}n Ugarte, editors, {\em 18th
  International Conference on Database Theory, {ICDT} 2015, March 23-27, 2015,
  Brussels, Belgium}, volume~31 of {\em LIPIcs}, pages 110--126. Schloss
  Dagstuhl - Leibniz-Zentrum f{\"{u}}r Informatik, 2015.
\newblock URL: \url{https://doi.org/10.4230/LIPIcs.ICDT.2015.110}, \href
  {https://doi.org/10.4230/LIPICS.ICDT.2015.110}
  {\path{doi:10.4230/LIPICS.ICDT.2015.110}}.

\bibitem{CoppersmithW82}
Don Coppersmith and Shmuel Winograd.
\newblock On the asymptotic complexity of matrix multiplication.
\newblock {\em {SIAM} J. Comput.}, 11(3):472--492, 1982.
\newblock \href {https://doi.org/10.1137/0211038} {\path{doi:10.1137/0211038}}.

\bibitem{CyganFKLMPPS15}
Marek Cygan, Fedor~V. Fomin, Lukasz Kowalik, Daniel Lokshtanov, D{\'{a}}niel
  Marx, Marcin Pilipczuk, Michal Pilipczuk, and Saket Saurabh.
\newblock {\em Parameterized Algorithms}.
\newblock Springer, 2015.
\newblock \href {https://doi.org/10.1007/978-3-319-21275-3}
  {\path{doi:10.1007/978-3-319-21275-3}}.

\bibitem{CzumajL09}
Artur Czumaj and Andrzej Lingas.
\newblock Finding a heaviest vertex-weighted triangle is not harder than matrix
  multiplication.
\newblock {\em {SIAM} J. Comput.}, 39(2):431--444, 2009.
\newblock \href {https://doi.org/10.1137/070695149}
  {\path{doi:10.1137/070695149}}.

\bibitem{DalirrooyfardVW21}
Mina Dalirrooyfard, Thuy{-}Duong Vuong, and Virginia {Vassilevska Williams}.
\newblock Graph pattern detection: Hardness for all induced patterns and faster
  noninduced cycles.
\newblock {\em {SIAM} J. Comput.}, 50(5):1627--1662, 2021.
\newblock \href {https://doi.org/10.1137/20M1335054}
  {\path{doi:10.1137/20M1335054}}.

\bibitem{DalmauJ04}
V{\'{\i}}ctor Dalmau and Peter Jonsson.
\newblock The complexity of counting homomorphisms seen from the other side.
\newblock {\em Theor. Comput. Sci.}, 329(1-3):315--323, 2004.
\newblock URL: \url{https://doi.org/10.1016/j.tcs.2004.08.008}, \href
  {https://doi.org/10.1016/J.TCS.2004.08.008}
  {\path{doi:10.1016/J.TCS.2004.08.008}}.

\bibitem{DeepHK22}
Shaleen Deep, Xiao Hu, and Paraschos Koutris.
\newblock Ranked enumeration of join queries with projections.
\newblock {\em Proc. {VLDB} Endow.}, 15(5):1024--1037, 2022.
\newblock URL: \url{https://www.vldb.org/pvldb/vol15/p1024-deep.pdf}, \href
  {https://doi.org/10.14778/3510397.3510401}
  {\path{doi:10.14778/3510397.3510401}}.

\bibitem{DellRW19}
Holger Dell, Marc Roth, and Philip Wellnitz.
\newblock Counting answers to existential questions.
\newblock In Christel Baier, Ioannis Chatzigiannakis, Paola Flocchini, and
  Stefano Leonardi, editors, {\em 46th International Colloquium on Automata,
  Languages, and Programming, {ICALP} 2019, July 9-12, 2019, Patras, Greece},
  volume 132 of {\em LIPIcs}, pages 113:1--113:15. Schloss Dagstuhl -
  Leibniz-Zentrum f{\"{u}}r Informatik, 2019.
\newblock URL: \url{https://doi.org/10.4230/LIPIcs.ICALP.2019.113}, \href
  {https://doi.org/10.4230/LIPICS.ICALP.2019.113}
  {\path{doi:10.4230/LIPICS.ICALP.2019.113}}.

\bibitem{DurandM13}
Arnaud Durand and Stefan Mengel.
\newblock Structural tractability of counting of solutions to conjunctive
  queries.
\newblock In Wang{-}Chiew Tan, Giovanna Guerrini, Barbara Catania, and
  Anastasios Gounaris, editors, {\em Joint 2013 {EDBT/ICDT} Conferences, {ICDT}
  '13 Proceedings, Genoa, Italy, March 18-22, 2013}, pages 81--92. {ACM}, 2013.
\newblock \href {https://doi.org/10.1145/2448496.2448508}
  {\path{doi:10.1145/2448496.2448508}}.

\bibitem{DurandM14}
Arnaud Durand and Stefan Mengel.
\newblock The complexity of weighted counting for acyclic conjunctive queries.
\newblock {\em J. Comput. Syst. Sci.}, 80(1):277--296, 2014.
\newblock URL: \url{https://doi.org/10.1016/j.jcss.2013.08.001}, \href
  {https://doi.org/10.1016/J.JCSS.2013.08.001}
  {\path{doi:10.1016/J.JCSS.2013.08.001}}.

\bibitem{EisenbrandG04}
Friedrich Eisenbrand and Fabrizio Grandoni.
\newblock On the complexity of fixed parameter clique and dominating set.
\newblock {\em Theor. Comput. Sci.}, 326(1-3):57--67, 2004.
\newblock URL: \url{https://doi.org/10.1016/j.tcs.2004.05.009}, \href
  {https://doi.org/10.1016/J.TCS.2004.05.009}
  {\path{doi:10.1016/J.TCS.2004.05.009}}.

\bibitem{FanKZ23}
Austen~Z. Fan, Paraschos Koutris, and Hangdong Zhao.
\newblock The fine-grained complexity of boolean conjunctive queries and
  sum-product problems.
\newblock In Kousha Etessami, Uriel Feige, and Gabriele Puppis, editors, {\em
  50th International Colloquium on Automata, Languages, and Programming,
  {ICALP} 2023, July 10-14, 2023, Paderborn, Germany}, volume 261 of {\em
  LIPIcs}, pages 127:1--127:20. Schloss Dagstuhl - Leibniz-Zentrum f{\"{u}}r
  Informatik, 2023.
\newblock URL: \url{https://doi.org/10.4230/LIPIcs.ICALP.2023.127}, \href
  {https://doi.org/10.4230/LIPICS.ICALP.2023.127}
  {\path{doi:10.4230/LIPICS.ICALP.2023.127}}.

\bibitem{Simons23}
Simons~Institute for the Theory~of Computing.
\newblock Logic and algorithms in database theory and ai boot camp.
\newblock
  \url{https://simons.berkeley.edu/workshops/logic-algorithms-database-theory-ai-boot-camp#simons-tabs},
  August 2023.

\bibitem{GajentaanO95}
Anka Gajentaan and Mark~H. Overmars.
\newblock On a class of {$O(n^2)$} problems in computational geometry.
\newblock {\em Comput. Geom.}, 5:165--185, 1995.
\newblock \href {https://doi.org/10.1016/0925-7721(95)00022-2}
  {\path{doi:10.1016/0925-7721(95)00022-2}}.

\bibitem{Gall12}
Fran{\c{c}}ois~Le Gall.
\newblock Faster algorithms for rectangular matrix multiplication.
\newblock In {\em 53rd Annual {IEEE} Symposium on Foundations of Computer
  Science, {FOCS} 2012, New Brunswick, NJ, USA, October 20-23, 2012}, pages
  514--523. {IEEE} Computer Society, 2012.
\newblock \href {https://doi.org/10.1109/FOCS.2012.80}
  {\path{doi:10.1109/FOCS.2012.80}}.

\bibitem{GrandjeanJ22}
Etienne Grandjean and Louis Jachiet.
\newblock Which arithmetic operations can be performed in constant time in the
  {RAM} model with addition?
\newblock {\em CoRR}, abs/2206.13851, 2022.
\newblock URL: \url{https://doi.org/10.48550/arXiv.2206.13851}, \href
  {https://arxiv.org/abs/2206.13851} {\path{arXiv:2206.13851}}, \href
  {https://doi.org/10.48550/ARXIV.2206.13851}
  {\path{doi:10.48550/ARXIV.2206.13851}}.

\bibitem{GrecoS14}
Gianluigi Greco and Francesco Scarcello.
\newblock Counting solutions to conjunctive queries: structural and hybrid
  tractability.
\newblock In Richard Hull and Martin Grohe, editors, {\em Proceedings of the
  33rd {ACM} {SIGMOD-SIGACT-SIGART} Symposium on Principles of Database
  Systems, PODS'14, Snowbird, UT, USA, June 22-27, 2014}, pages 132--143.
  {ACM}, 2014.
\newblock \href {https://doi.org/10.1145/2594538.2594559}
  {\path{doi:10.1145/2594538.2594559}}.

\bibitem{Grohe02}
Martin Grohe.
\newblock Parameterized complexity for the database theorist.
\newblock {\em {SIGMOD} Rec.}, 31(4):86--96, 2002.
\newblock \href {https://doi.org/10.1145/637411.637428}
  {\path{doi:10.1145/637411.637428}}.

\bibitem{Grohe07}
Martin Grohe.
\newblock The complexity of homomorphism and constraint satisfaction problems
  seen from the other side.
\newblock {\em J. {ACM}}, 54(1):1:1--1:24, 2007.
\newblock \href {https://doi.org/10.1145/1206035.1206036}
  {\path{doi:10.1145/1206035.1206036}}.

\bibitem{GroheSS01}
Martin Grohe, Thomas Schwentick, and Luc Segoufin.
\newblock When is the evaluation of conjunctive queries tractable?
\newblock In Jeffrey~Scott Vitter, Paul~G. Spirakis, and Mihalis Yannakakis,
  editors, {\em Proceedings on 33rd Annual {ACM} Symposium on Theory of
  Computing, July 6-8, 2001, Heraklion, Crete, Greece}, pages 657--666. {ACM},
  2001.
\newblock \href {https://doi.org/10.1145/380752.380867}
  {\path{doi:10.1145/380752.380867}}.

\bibitem{Hu24}
Xiao Hu.
\newblock Fast matrix multiplication for query processing.
\newblock {\em Proc. {ACM} Manag. Data}, 2(2):98, 2024.
\newblock \href {https://doi.org/10.1145/3651599} {\path{doi:10.1145/3651599}}.

\bibitem{ImpagliazzoP01}
Russell Impagliazzo and Ramamohan Paturi.
\newblock On the complexity of k-sat.
\newblock {\em J. Comput. Syst. Sci.}, 62(2):367--375, 2001.
\newblock URL: \url{https://doi.org/10.1006/jcss.2000.1727}, \href
  {https://doi.org/10.1006/JCSS.2000.1727} {\path{doi:10.1006/JCSS.2000.1727}}.

\bibitem{ImpagliazzoPZ01}
Russell Impagliazzo, Ramamohan Paturi, and Francis Zane.
\newblock Which problems have strongly exponential complexity?
\newblock {\em J. Comput. Syst. Sci.}, 63(4):512--530, 2001.
\newblock URL: \url{https://doi.org/10.1006/jcss.2001.1774}, \href
  {https://doi.org/10.1006/JCSS.2001.1774} {\path{doi:10.1006/JCSS.2001.1774}}.

\bibitem{KaraNNOZ19}
Ahmet Kara, Hung~Q. Ngo, Milos Nikolic, Dan Olteanu, and Haozhe Zhang.
\newblock Counting triangles under updates in worst-case optimal time.
\newblock In Pablo Barcel{\'{o}} and Marco Calautti, editors, {\em 22nd
  International Conference on Database Theory, {ICDT} 2019, March 26-28, 2019,
  Lisbon, Portugal}, volume 127 of {\em LIPIcs}, pages 4:1--4:18. Schloss
  Dagstuhl - Leibniz-Zentrum f{\"{u}}r Informatik, 2019.
\newblock URL: \url{https://doi.org/10.4230/LIPIcs.ICDT.2019.4}, \href
  {https://doi.org/10.4230/LIPICS.ICDT.2019.4}
  {\path{doi:10.4230/LIPICS.ICDT.2019.4}}.

\bibitem{KaraNOZ20}
Ahmet Kara, Milos Nikolic, Dan Olteanu, and Haozhe Zhang.
\newblock Trade-offs in static and dynamic evaluation of hierarchical queries.
\newblock In Dan Suciu, Yufei Tao, and Zhewei Wei, editors, {\em Proceedings of
  the 39th {ACM} {SIGMOD-SIGACT-SIGAI} Symposium on Principles of Database
  Systems, {PODS} 2020, Portland, OR, USA, June 14-19, 2020}, pages 375--392.
  {ACM}, 2020.
\newblock \href {https://doi.org/10.1145/3375395.3387646}
  {\path{doi:10.1145/3375395.3387646}}.

\bibitem{KaraNOZ23}
Ahmet Kara, Milos Nikolic, Dan Olteanu, and Haozhe Zhang.
\newblock Conjunctive queries with free access patterns under updates.
\newblock In Floris Geerts and Brecht Vandevoort, editors, {\em 26th
  International Conference on Database Theory, {ICDT} 2023, March 28-31, 2023,
  Ioannina, Greece}, volume 255 of {\em LIPIcs}, pages 17:1--17:20. Schloss
  Dagstuhl - Leibniz-Zentrum f{\"{u}}r Informatik, 2023.
\newblock URL: \url{https://doi.org/10.4230/LIPIcs.ICDT.2023.17}, \href
  {https://doi.org/10.4230/LIPICS.ICDT.2023.17}
  {\path{doi:10.4230/LIPICS.ICDT.2023.17}}.

\bibitem{Karp72}
Richard~M. Karp.
\newblock Reducibility among combinatorial problems.
\newblock In Raymond~E. Miller and James~W. Thatcher, editors, {\em Proceedings
  of a symposium on the Complexity of Computer Computations, held March 20-22,
  1972, at the {IBM} Thomas J. Watson Research Center, Yorktown Heights, New
  York, {USA}}, The {IBM} Research Symposia Series, pages 85--103. Plenum
  Press, New York, 1972.
\newblock \href {https://doi.org/10.1007/978-1-4684-2001-2\_9}
  {\path{doi:10.1007/978-1-4684-2001-2\_9}}.

\bibitem{KhamisCMNNOS20}
Mahmoud~Abo Khamis, Ryan~R. Curtin, Benjamin Moseley, Hung~Q. Ngo, XuanLong
  Nguyen, Dan Olteanu, and Maximilian Schleich.
\newblock Functional aggregate queries with additive inequalities.
\newblock {\em {ACM} Trans. Database Syst.}, 45(4):17:1--17:41, 2020.
\newblock \href {https://doi.org/10.1145/3426865} {\path{doi:10.1145/3426865}}.

\bibitem{AboKhamisHS24}
Mahmoud~Abo Khamis, Xiao Hu, and Dan Suciu.
\newblock Fast matrix multiplication meets the subdmodular width.
\newblock {\em CoRR}, abs/2412.06189, 2024.
\newblock URL: \url{https://doi.org/10.48550/arXiv.2412.06189}, \href
  {https://arxiv.org/abs/2412.06189} {\path{arXiv:2412.06189}}, \href
  {https://doi.org/10.48550/ARXIV.2412.06189}
  {\path{doi:10.48550/ARXIV.2412.06189}}.

\bibitem{KhamisNR16}
Mahmoud~Abo Khamis, Hung~Q. Ngo, and Atri Rudra.
\newblock {FAQ:} questions asked frequently.
\newblock In Tova Milo and Wang{-}Chiew Tan, editors, {\em Proceedings of the
  35th {ACM} {SIGMOD-SIGACT-SIGAI} Symposium on Principles of Database Systems,
  {PODS} 2016, San Francisco, CA, USA, June 26 - July 01, 2016}, pages 13--28.
  {ACM}, 2016.
\newblock \href {https://doi.org/10.1145/2902251.2902280}
  {\path{doi:10.1145/2902251.2902280}}.

\bibitem{Khamis0S17}
Mahmoud~Abo Khamis, Hung~Q. Ngo, and Dan Suciu.
\newblock What do shannon-type inequalities, submodular width, and disjunctive
  datalog have to do with one another?
\newblock In Emanuel Sallinger, Jan~Van den Bussche, and Floris Geerts,
  editors, {\em Proceedings of the 36th {ACM} {SIGMOD-SIGACT-SIGAI} Symposium
  on Principles of Database Systems, {PODS} 2017, Chicago, IL, USA, May 14-19,
  2017}, pages 429--444. {ACM}, 2017.
\newblock \href {https://doi.org/10.1145/3034786.3056105}
  {\path{doi:10.1145/3034786.3056105}}.

\bibitem{LincolnWW18}
Andrea Lincoln, Virginia {Vassilevska Williams}, and R.~Ryan Williams.
\newblock Tight hardness for shortest cycles and paths in sparse graphs.
\newblock In Artur Czumaj, editor, {\em Proceedings of the Twenty-Ninth Annual
  {ACM-SIAM} Symposium on Discrete Algorithms, {SODA} 2018, New Orleans, LA,
  USA, January 7-10, 2018}, pages 1236--1252. {SIAM}, 2018.
\newblock \href {https://doi.org/10.1137/1.9781611975031.80}
  {\path{doi:10.1137/1.9781611975031.80}}.

\bibitem{Marx13}
D{\'{a}}niel Marx.
\newblock Tractable hypergraph properties for constraint satisfaction and
  conjunctive queries.
\newblock {\em J. {ACM}}, 60(6):42:1--42:51, 2013.
\newblock \href {https://doi.org/10.1145/2535926} {\path{doi:10.1145/2535926}}.

\bibitem{Mengel21}
Stefan Mengel.
\newblock A short note on the counting complexity of conjunctive queries.
\newblock {\em CoRR}, abs/2112.01108, 2021.
\newblock URL: \url{https://arxiv.org/abs/2112.01108}, \href
  {https://arxiv.org/abs/2112.01108} {\path{arXiv:2112.01108}}.

\bibitem{NesetrilP85}
Jaroslav Ne{\v{s}}et{\v{r}}il and Svatopluk Poljak.
\newblock On the complexity of the subgraph problem.
\newblock {\em Commentationes Mathematicae Universitatis Carolinae},
  26(2):415--419, 1985.

\bibitem{Ngo18}
Hung~Q. Ngo.
\newblock Worst-case optimal join algorithms: Techniques, results, and open
  problems.
\newblock In Jan~Van den Bussche and Marcelo Arenas, editors, {\em Proceedings
  of the 37th {ACM} {SIGMOD-SIGACT-SIGAI} Symposium on Principles of Database
  Systems, Houston, TX, USA, June 10-15, 2018}, pages 111--124. {ACM}, 2018.
\newblock \href {https://doi.org/10.1145/3196959.3196990}
  {\path{doi:10.1145/3196959.3196990}}.

\bibitem{NgoPRR18}
Hung~Q. Ngo, Ely Porat, Christopher R{\'{e}}, and Atri Rudra.
\newblock Worst-case optimal join algorithms.
\newblock {\em J. {ACM}}, 65(3):16:1--16:40, 2018.
\newblock \href {https://doi.org/10.1145/3180143} {\path{doi:10.1145/3180143}}.

\bibitem{OlteanuZ15}
Dan Olteanu and Jakub Z{\'{a}}vodn{\'{y}}.
\newblock Size bounds for factorised representations of query results.
\newblock {\em {ACM} Trans. Database Syst.}, 40(1):2:1--2:44, 2015.
\newblock \href {https://doi.org/10.1145/2656335} {\path{doi:10.1145/2656335}}.

\bibitem{PatrascuW10}
Mihai P{\u{a}}tra{\c{s}}cu and Ryan Williams.
\newblock On the possibility of faster {SAT} algorithms.
\newblock In Moses Charikar, editor, {\em Proceedings of the Twenty-First
  Annual {ACM-SIAM} Symposium on Discrete Algorithms, {SODA} 2010, Austin,
  Texas, USA, January 17-19, 2010}, pages 1065--1075. {SIAM}, 2010.
\newblock \href {https://doi.org/10.1137/1.9781611973075.86}
  {\path{doi:10.1137/1.9781611973075.86}}.

\bibitem{PichlerS13}
Reinhard Pichler and Sebastian Skritek.
\newblock Tractable counting of the answers to conjunctive queries.
\newblock {\em J. Comput. Syst. Sci.}, 79(6):984--1001, 2013.
\newblock URL: \url{https://doi.org/10.1016/j.jcss.2013.01.012}, \href
  {https://doi.org/10.1016/J.JCSS.2013.01.012}
  {\path{doi:10.1016/J.JCSS.2013.01.012}}.

\bibitem{Segoufin14}
Luc Segoufin.
\newblock A glimpse on constant delay enumeration (invited talk).
\newblock In Ernst~W. Mayr and Natacha Portier, editors, {\em 31st
  International Symposium on Theoretical Aspects of Computer Science {(STACS}
  2014), {STACS} 2014, March 5-8, 2014, Lyon, France}, volume~25 of {\em
  LIPIcs}, pages 13--27. Schloss Dagstuhl - Leibniz-Zentrum f{\"{u}}r
  Informatik, 2014.
\newblock URL: \url{https://doi.org/10.4230/LIPIcs.STACS.2014.13}, \href
  {https://doi.org/10.4230/LIPICS.STACS.2014.13}
  {\path{doi:10.4230/LIPICS.STACS.2014.13}}.

\bibitem{Segoufin15}
Luc Segoufin.
\newblock Constant delay enumeration for conjunctive queries.
\newblock {\em {SIGMOD} Rec.}, 44(1):10--17, 2015.
\newblock \href {https://doi.org/10.1145/2783888.2783894}
  {\path{doi:10.1145/2783888.2783894}}.

\bibitem{Strassen69}
Volker Strassen.
\newblock Gaussian elimination is not optimal.
\newblock {\em Numerische mathematik}, 13(4):354--356, 1969.

\bibitem{TziavelisAGRY20}
Nikolaos Tziavelis, Deepak Ajwani, Wolfgang Gatterbauer, Mirek Riedewald, and
  Xiaofeng Yang.
\newblock Optimal algorithms for ranked enumeration of answers to full
  conjunctive queries.
\newblock {\em Proc. {VLDB} Endow.}, 13(9):1582--1597, 2020.
\newblock URL: \url{http://www.vldb.org/pvldb/vol13/p1582-tziavelis.pdf}, \href
  {https://doi.org/10.14778/3397230.3397250}
  {\path{doi:10.14778/3397230.3397250}}.

\bibitem{Valiant75}
Leslie~G. Valiant.
\newblock General context-free recognition in less than cubic time.
\newblock {\em J. Comput. Syst. Sci.}, 10(2):308--315, 1975.
\newblock \href {https://doi.org/10.1016/S0022-0000(75)80046-8}
  {\path{doi:10.1016/S0022-0000(75)80046-8}}.

\bibitem{Vardi82}
Moshe~Y. Vardi.
\newblock The complexity of relational query languages (extended abstract).
\newblock In Harry~R. Lewis, Barbara~B. Simons, Walter~A. Burkhard, and
  Lawrence~H. Landweber, editors, {\em Proceedings of the 14th Annual {ACM}
  Symposium on Theory of Computing, May 5-7, 1982, San Francisco, California,
  {USA}}, pages 137--146. {ACM}, 1982.
\newblock \href {https://doi.org/10.1145/800070.802186}
  {\path{doi:10.1145/800070.802186}}.

\bibitem{WilliamsW18}
Virginia {Vassilevska Williams} and R.~Ryan Williams.
\newblock Subcubic equivalences between path, matrix, and triangle problems.
\newblock {\em J. {ACM}}, 65(5):27:1--27:38, 2018.
\newblock \href {https://doi.org/10.1145/3186893} {\path{doi:10.1145/3186893}}.

\bibitem{WilliamsW13}
Virginia {Vassilevska Williams} and Ryan Williams.
\newblock Finding, minimizing, and counting weighted subgraphs.
\newblock {\em {SIAM} J. Comput.}, 42(3):831--854, 2013.
\newblock \href {https://doi.org/10.1137/09076619X}
  {\path{doi:10.1137/09076619X}}.

\bibitem{VassilevskaWilliams2018}
Virginia~Vassilevska Williams.
\newblock On some fine-grained questions in algorithms and complexity.
\newblock In {\em Proceedings of the international congress of mathematicians:
  Rio de janeiro 2018}, pages 3447--3487. World Scientific, 2018.

\bibitem{Yannakakis81}
Mihalis Yannakakis.
\newblock Algorithms for acyclic database schemes.
\newblock In {\em Very Large Data Bases, 7th International Conference,
  September 9-11, 1981, Cannes, France, Proceedings}, pages 82--94. {IEEE}
  Computer Society, 1981.

\bibitem{YusterZ04}
Raphael Yuster and Uri Zwick.
\newblock Detecting short directed cycles using rectangular matrix
  multiplication and dynamic programming.
\newblock In J.~Ian Munro, editor, {\em Proceedings of the Fifteenth Annual
  {ACM-SIAM} Symposium on Discrete Algorithms, {SODA} 2004, New Orleans,
  Louisiana, USA, January 11-14, 2004}, pages 254--260. {SIAM}, 2004.
\newblock URL: \url{http://dl.acm.org/citation.cfm?id=982792.982828}.

\bibitem{ZhaoDK23}
Hangdong Zhao, Shaleen Deep, and Paraschos Koutris.
\newblock Space-time tradeoffs for conjunctive queries with access patterns.
\newblock In Floris Geerts, Hung~Q. Ngo, and Stavros Sintos, editors, {\em
  Proceedings of the 42nd {ACM} {SIGMOD-SIGACT-SIGAI} Symposium on Principles
  of Database Systems, {PODS} 2023, Seattle, WA, USA, June 18-23, 2023}, pages
  59--68. {ACM}, 2023.
\newblock \href {https://doi.org/10.1145/3584372.3588675}
  {\path{doi:10.1145/3584372.3588675}}.

\bibitem{ZhaoFOK24}
Hangdong Zhao, Austen~Z. Fan, Xiating Ouyang, and Paraschos Koutris.
\newblock Conjunctive queries with negation and aggregation: {A} linear time
  characterization.
\newblock {\em Proc. {ACM} Manag. Data}, 2(2):75, 2024.
\newblock \href {https://doi.org/10.1145/3651138} {\path{doi:10.1145/3651138}}.

\end{thebibliography}
          
\end{document}